\documentclass[letterpaper, 11pt]{article}

\usepackage[margin=1in]{geometry}
\usepackage[T1]{fontenc}
\usepackage{lmodern}
\usepackage{microtype}

\usepackage{amsmath,amssymb,amsfonts, amsthm}
\usepackage{mathtools}

\theoremstyle{plain}
\newtheorem{theorem}{Theorem}[section]
\newtheorem{proposition}[theorem]{Proposition}

\theoremstyle{definition}
\newtheorem{definition}[theorem]{Definition}

\providecommand{\mathbbm}[1]{\mathbf{#1}}

\usepackage{hyperref}

\hypersetup{bookmarksdepth=3, bookmarksnumbered=true, bookmarksopen=true,
  bookmarksopenlevel=1, pdfpagemode=UseOutlines}
\usepackage{booktabs}
\usepackage{etoc}
\AtBeginDocument{\etocdepthtag.toc{mainmatter}}
\usepackage[inline]{enumitem}
\usepackage{algorithm}
\usepackage{algpseudocode}
\usepackage[lambda,advantage,operators,sets,adversary,landau,probability,notions,logic,ff,mm,primitives,events,complexity,asymptotics,keys]{cryptocode}

\usepackage[capitalize]{cleveref}

\newcommand{\multilinePr}[2]{\Pr\left[\begin{aligned}#1\end{aligned} \;:\; \begin{aligned}#2\end{aligned}\right]}

\newcommand{\setup}{\mathsf{Setup}}
\newcommand{\KG}{\mathsf{Keygen}}
\newcommand{\prv}{\mathsf{Prove}}
\newcommand{\vfy}{\mathsf{Vfy}}

\newcommand{\pcrs}{\mathsf{crs}_P}
\newcommand{\vcrs}{\mathsf{crs}_V}

\newcommand{\rel}{\mathcal{R}}
\newcommand{\xwDist}{\mathcal{D}}

\newcommand{\ro}{\mathcal{H}}
\newcommand{\roSet}{\Omega_\lambda}
\newcommand{\bits}{\left\{ 0,1 \right\}}

\newcommand{\polylog}{\mathrm{polylog}}

\renewcommand{\negl}{\mathsf{negl}(\lambda)}

\newcommand{\Com}{\mathsf{Com}}
\newcommand{\cm}{\mathsf{cm}}
\newcommand{\nonce}{\mathsf{nonce}}

\begin{document}

\title{One Node, Two Roles\\[0.4em]
  {\large Simultaneous Contests for Validation and Attention in Rollups}}
\author{
  Pranay Anchuri\\
  {\small Offchain Labs}\\[-2pt]
  {\small\texttt{panchuri@offchainlabs.com}}
  \and
  Ben Berger\\
  {\small Offchain Labs}\\[-2pt]
  {\small\texttt{bberger@offchainlabs.com}}
  \and
  Matteo Campanelli\\
  {\small Offchain Labs \& University of Tartu}\\[-2pt]
  {\small\texttt{matteo@offchainlabs.com}}
  \and
  Akaki Mamageishvili\\
  {\small Offchain Labs}\\[-2pt]
  {\small\texttt{amamageishvili@offchainlabs.com}}
}
\date{}
\maketitle

\begin{abstract}

An optimistic rollup is safe as long as at least one honest validator executes its state transition function (STF) and disputes any assertion that is inconsistent with the result of the execution. Allowing anyone to participate, however, does not
give incentive to do so: as long as every assertion is correct, verifying assertion correctness
goes unpaid. 

Attention mechanisms address this by paying nodes to execute the rollup's STF, regardless of whether they participate in the validation process. Since a single node
operator can back many registered identities with a single execution, paying
per identity does not buy execution diversity, that is, the number of operators
that independently execute.

In this work we provide a new modeling framework for this problem.
We formalize the underlying cryptographic primitive as a new notion, arguments with registered provers, whose properties include a form of \emph{non-amortizability} that, to our knowledge, has not been studied for modern succinct cryptographic proofs before.

We design validation and attention reward mechanisms and analyze them as two coupled
Tullock-like contests, in which an operator pays once for the required execution and then chooses how
many identities to deploy in each role. 
We quantify the attainable range of execution diversities in equilibrium as a function of the marginal cost of deploying an extra attention identity.

We use our framework as a lens through which we compare
two attention mechanisms, TRACE (MARBLE 2026) and Proof of Diligence (AFT 2024).
In particular, we show that as opposed to Proof of Diligence, in TRACE the marginal cost mentioned above can be adjusted to support a higher execution diversity.

As a case study, we apply our results to Arbitrum One, a widely deployed optimistic rollup.
\end{abstract}

\section{Introduction}
\label{sec:intro}
Optimistic rollups rely on designated parties, named \emph{validators}, to periodically post commitments to snapshots of the rollup's state on the parent chain.
The safety of these rollups is achieved through dispute resolution protocols which ensure that only commitments, named hereafter \emph{assertions}, that correspond to the correct state get accepted by the protocol as canonical.

One property of these protocols is the existence of a challenge window --- usually set at one week --- in which validators can dispute the correctness of a posted assertion by posting a competing assertion alongside some predetermined stake.
Furthermore, these protocols are designed so that only the unique correct assertion, if indeed posted by some validator, is guaranteed to get accepted.
This establishes a very desirable 1-of-N safety assumption for optimistic rollups, whereby only a single honest validator is required to guarantee the rollup's safety.

However, the 1-of-N assumption crucially requires that some party be \emph{willing to act} as the honest validator. Doing so is costly: the party must execute the rollup's state transition function (hereafter STF) to compute the correct assertion and must lock up the required stake, which in the case of Arbitrum's dispute protocol BoLD~\cite{bold} is at least 3600 ETH~\cite{arbitrum-aip-bold}.
Furthermore, in case of a dispute honest validators need to be prepared to submit all required protocol moves to guarantee honest confirmation.
Finally, validators are typically not compensated for their continuous work as long as there is no incorrect assertion to challenge.

The two largest optimistic rollups, Base and Arbitrum, have made the validator role open to anyone. However, allowing anyone to participate does not give incentive to do so.  Indeed, the validator role requires expertise which, along with the combined cost of execution and the locked stake, renders the role unsuitable for most parties.
This is the verifier's dilemma of Luu et al.~\cite{verifiersdilemma}
in its rollup form.

\paragraph{Attention Nodes}
To mitigate the complexity of rollup validation, previous work \cite{mamageishvili2023, pod, trace} has proposed rewarding designated \emph{attention nodes} for executing the STF, regardless of whether they participate in the validation protocol.

Rewarding attention nodes is desirable for two reasons. First, an attention node's work is a prerequisite for validation, so a larger pool of attention nodes, any of which could step in as a validator, deters malicious parties from posting false assertions.
Second, attention nodes can provide necessary block data to light clients and wallets~\cite{frauddataavailability}.

\paragraph{Execution Diversity, the Sybil Problem
and Isolating Attention from Validation}
Notably, the benefits of attention as described above grow with the number of independently operated nodes, not with
the number of identities registered as attention nodes behind a single executing node.
In the Proof-of-Diligence mechanism~\cite{pod} and the TRACE mechanism~\cite{trace}, each registered identity submits a \emph{proof of execution}, which one of the rollup smart contracts verifies.  

However, as nothing ties an identity to a distinct machine, a single operator can register many independent identities which in turn submit their proofs based on the same execution trace.
We hereafter refer to these identities as \emph{sybils}, and we call the number of distinct operators independently executing the rollup's STF the \emph{execution diversity} of the rollup.

Proof-of-Diligence~\cite{pod} requires each registered identity to lock stake in order to alleviate the sybil problem.
Furthermore, they show that with a certain whistleblower mechanism independent node operators do not pool their execution efforts in equilibrium.
However, the obtained execution diversity is not analyzed and remains an open question.

Moreover, prior work analyzes attention in isolation, without modeling how taking on one role affects the incentive to take on the other. In practice, the two roles are coupled: both build on the same execution of the STF. Once an operator executes for one role, that cost is sunk, and taking on the other requires only the additional per-identity cost. Yet, to our knowledge, no prior work analyzes the economics of attention and validation jointly.

\subsection{Our Contributions}
\label{sec:contributions}

\paragraph{Cost conditions for execution diversity}
We analyze lottery-based, reward mechanisms for attention and validation. A node operator pays a single entry cost to participate in one or both roles, corresponding to the cost of obtaining the execution trace. Each role then adds a marginal cost per registered sybil and pays out a total reward that depends on the number of registered sybils. Assuming risk-neutral operators, we model their utilities by their expected payout, so that every sybil receives, in expectation, an equal share of its role's reward.  The resulting game resembles two simultaneous Tullock-like contests coupled through the shared entry cost.
We derive conditions on the entry and marginal sybil costs that support a given target execution diversity in equilibrium.

\paragraph{A registered prover model for attention mechanisms} 
We introduce arguments with registered provers, a cryptographic framework for attention mechanisms, and formalize
\textit{non-amortizability}, which captures the per-identity work needed by our game-theoretic analysis.
We use this framework to compare the cryptographic guarantees of TRACE~\cite{trace} and Proof of Diligence~\cite{pod}.

As an additional side contribution, we sketch a construction with stronger non-amortizability guarantees than both TRACE and Proof of Diligence.  

The construction, however, requires a more expensive prover compared to prior schemes. 
We leave a concrete cost estimate for this construction to future work.
\paragraph{A worked example on Arbitrum One}
We apply our theoretical results to the case of Arbitrum One.
Our measurements put the attention sybil cost at $d_a = \$0.91$ per hour at the default trace length of $N = 2^{20}$. At this cost, non-decreasing reward functions support an execution diversity of at most $A = 2$. We show that increasing the trace length to $N = 2^{24}$ raises this maximum to $A = 29$, while reward functions that decrease beyond some threshold support $A = 30$ at the baseline cost. We discuss the resulting trade-offs between costs and rewards.

\paragraph{Roadmap.}
The rest of the paper is organized as follows.
\Cref{sec:overview} gives an overview of our results.
\Cref{sec:mechanism} defines arguments with registered provers, casts TRACE as
one, and states its non-amortizability.
\Cref{sec:game} develops the two-contest model; \cref{sec:example} instantiates it
on Arbitrum One.
We discuss our additional construction in \Cref{sec:variant-limitation}  and review related work in \cref{sec:related}.

\section{Overview}
\label{sec:overview}

An optimistic rollup applies a batch of transactions $T$ to a prior state $S'$
and posts an assertion about the resulting state $S$ on the settlement chain.
Executing the batch also produces an \emph{execution trace} $E$, the sequence of
intermediate values held during execution, whose length $N$ the deployer fixes
when instrumenting the state transition function.\footnote{$E$ is not the full
step-by-step computation history that a validity rollup commits to when proving
the state transition. The mechanisms we study prove a statement about a vector of
$N$ values derived from the execution, and $N$ is a parameter of the attention
mechanism rather than a property of the batch.}
Standard node interfaces do not expose it, so a node that skipped execution does
not have it, while the state transition function fixes a commitment to it that
the chain already stores.
The mechanisms we study pay an identity only for a \emph{claim}, which is what an
identity submits when its proof passes an eligibility threshold.
That threshold test is the lottery of \cref{sec:contributions}.
\Cref{sec:game} does not model individual draws and works with the expected
payout, in which each sybil earns an equal share of its role's reward.
Both mechanisms also assume that an operator which has not executed the batch
does not obtain $E$ from one that has.
Otherwise a single execution could supply claims for many operators, and the
reward would buy claims rather than executions.
Neither rules this out cryptographically.
TRACE argues that buying the trace is unprofitable once the bandwidth cost of
transferring it exceeds the expected reward~\cite{trace}.
Proof of Diligence treats resale as collusion, which we discuss below.

\paragraph{The game.}
\Cref{sec:game} models how operators respond to a reward paid per sybil.
An operator pays a cost $c$ once to run a node, each attention sybil it
registers costs a further $d_a$, and the reward is split among sybils in
proportion to how many each operator controls.
The outcome a rollup wants is one identity per operator, as this would result in the minimal offered reward for otherwise the same amount of independent operators.
When the total reward does not decrease as identities are added, this outcome is
an equilibrium only if the number of attention operators satisfies
\[
  A \le 1 + \frac{2d_a}{c}.
\]
The larger the per-identity cost $d_a$ relative to the node cost $c$, the more
independent operators the reward can sustain.
When an extra identity is cheap, an existing operator captures additional reward
more cheaply by registering another identity than a newcomer can by running
another node, so the reward buys identities rather than operators.
A reward function that decreases past a target avoids this issue, whereas under a
non-decreasing reward, $d_a$ must be adjusted to support more independent attention node operators in equlibrium.

\paragraph{Modeling the per-identity cost.}
To reason about $d_a$, \cref{sec:mechanism} introduces a new notion of cryptographic proofs, \textit{arguments with registered
provers}.
In it, each identity holds a registered key pair. An identity holder can produce a certificate under its key, a \textit{proof},  ensuring 
that the prover knows the execution trace behind a commitment, in the context of
a given block. 

The property that connects this notion to $d_a$ is non-amortizability
(\cref{def:non-amort}).
Execution is paid once per node and shared by all of an operator's identities,
so only the work of producing a claim
can contribute to $d_a$, and only to the
extent that an operator cannot share it across its keys.
Non-amortizability with overhead $\Delta$ says that producing valid proofs under
$n$ distinct keys takes time roughly $n\Delta$, so $\Delta$ is a lower bound on the
proving component contributing to $d_a$ that no batching across identities removes.

\paragraph{TRACE in this framework.}
We adopt a variant of TRACE~\cite{trace}. Here, a node that has executed a block proves,
with a succinct inner-product argument, that it knows the trace behind the
commitment the chain already stores; it then signs the proof with its registered
key.
In TRACE, reward eligibility is determined from the proof, so an identity incurs the proving cost before learning whether it qualifies for a reward.
We cast our main variant of TRACE as an argument with registered provers
(\cref{sec:ipa-construction}) and show that it
is non-amortizable with overhead
$\Delta=\Omega(\log N\cdot\lambda^{1-\epsilon})$, where $N$ is the trace length and $\lambda$ is the security parameter
(\cref{thm:ipa-nonamort,tab:properties}).
This overhead stems from the argument system underlying TRACE, a Bulletproofs-like inner-product argument~\cite{bootle2016ipa,bulletproofs}, which we instantiate with the publicly verifiable variant of Daza et al.~\cite{drz}.
Its prover runs $\log N$ rounds, each halving the committed vectors by folding them with a Fiat--Shamir challenge.
Since the public key seeds these challenges, every folding step yields a key-dependent combination of the entire trace, on which the next round computes group elements that no other identity can reuse.

\paragraph{Implications for TRACE and Proof of Diligence.}
We then use the framework to compare TRACE with Proof of
Diligence~\cite{pod} (\cref{tab:properties}).
In Proof of Diligence, each identity evaluates a verifiable random function
on the Merkle roots of the new state and of the execution trace, and it is
paid if the output falls below a threshold that depends on its stake.
The proof certifies the roots, not the trace, and its input is shared by every
identity of an operator, so Proof of Diligence has $\Delta=O(1)$.
Its per-identity cost is therefore a constant that the deployer cannot raise,
and under a non-decreasing reward the bound on $A$ stays fixed.
In TRACE, $\Delta$ grows with the trace length, which makes the trace length a
parameter that raises $d_a$ and with it the bound.
On Arbitrum One this gives $A\le 2$ at trace length $2^{20}$ and $A\le 29$ at
$2^{24}$ (\cref{sec:example}).
The guarantee is weaker than the honest prover's cost: $\Delta$ grows only
logarithmically in $N$ while the honest prover does work which is linear in $N$.
\Cref{alg:both} shows the two claim procedures side by side.

\paragraph{Additional observations on Proof of Diligence.}
Besides allowing nodes to amortize costs, Proof of Diligence does not guarantee that a
claimant knows the trace.
Its proof certifies the root $r_E$, so a node that obtains the correct root from
another node produces an accepting proof without having executed the block.
A mechanism whose proofs are knowledge sound in the sense of \cref{def:ks}, as
TRACE's are~\cite{trace}, rules this out, since an accepting proof lets an
extractor recover the trace itself.
Proof of Diligence does not completely ignore the issue, 
but it addresses it with an
incentive-based approach rather than a cryptographic-based argument.
It assumes that $r_E$ can be obtained only by executing, lets other watchtowers recompute $r_E$ and challenge an incorrect proof during a challenge period of the watchtower protocol (distinct from the rollup's dispute window), and
treats selling a computed $r_E$ as a collusion strategy; with an added
whistleblower contract, independent diligence becomes the unique
equilibrium~\cite[Cor.~8]{pod}.
That guarantee holds for rational and independent watchtowers.
However, it does not apply to an operator that computes $r_E$ once and uses it for all its own
identities, which the analysis counts as diligence rather than collusion.
As we argue, however, this behavior is undesired and the current Proof of Diligence construction cannot provide the stronger property we believe is necessary in an attention mechanism, namely, attention claims that guarantee knowledge of the entire execution trace.

\paragraph{A construction with stronger non-amortizability.}
The logarithmic overhead of TRACE comes from the fact that not all inputs to
the inner-product argument have high computational entropy.
In \cref{sec:prf-cons}, we instead derive a secret vector from a registered PRF
key and the block, so every identity folds the trace against fresh pseudorandom
values on every block, and a general-purpose commit-and-prove argument shows that
the vector was derived correctly.
Each identity then incurs fresh work proportional to the trace length, which raises the
overhead to linear in the trace length, within a logarithmic factor of the honest
prover's cost.
This stronger dependence of $\Delta$ on trace length is desirable because trace
length is our main lever for raising the unavoidable per-identity proving cost
and hence $d_a$.
The argument treats the PRF as a random oracle, because the adversary of
\cref{def:non-amort} holds every key, and the construction pays for the stronger
bound with a more expensive prover.

\begin{algorithm}[t]
\caption{Claim procedures for one block and one registered identity $i$.
}
\label{alg:both}
\begin{minipage}[t]{0.46\textwidth}
    \fussy
\textbf{Proof of Diligence}~\cite[Alg.~1]{pod} on $(S', T, r_S, \alpha_i)$
\begin{algorithmic}[1]
  \State $(S, E) \gets \mathsf{apply}(S', T)$
  \State $r'_S \gets \mathsf{Merklize}(S)$
  \State $r_E \gets \mathsf{Merklize}(E)$
  \State $(d, \pi) \gets \mathsf{VRF}_{\mathit{sk}_i}(r'_S \,\|\, r_E)$
  \State \textbf{if} $d < \phi(\alpha_i)$ \textbf{then submit} $(d, \pi)$
  \State \textbf{if} $r_S \neq r'_S$ \textbf{then raise alarm}
\end{algorithmic}
\end{minipage}\hfill
\begin{minipage}[t]{0.48\textwidth}
\textbf{TRACE} (\cref{sec:ipa-construction}) on $(S', T, \mathit{sk}_i)$
\begin{algorithmic}[1]
  \State $(S, E) \gets \mathsf{apply}(S', T)$
  \State $C_E \gets \mathsf{Commit}(E)$ \Comment{from chain}
  \State $v \gets \textstyle\sum_j E[j]$
  \State $\theta \gets H(C_E, \mathit{pk}_i, v, \mathsf{blockhash})$
  \State $\pi \gets \mathsf{IPA.Prove}\bigl(C_E, [S]_1, v;\, E, \mathbf 1;\, \theta\bigr)$
  \State $\sigma \gets \mathsf{Sign}_{\mathit{sk}_i}(C_E, \mathsf{blockhash}, v, \pi)$
  \State \textbf{if} $\mathsf{isReward}(C_E, \pi) \le 2^{\lambda - \tau}$
         \textbf{then submit} $(v, \pi, \sigma)$
\end{algorithmic}
\end{minipage}
\end{algorithm}

\subsection{Related Work}
\label{sec:related}

The game described in Section~\ref{sec:game} has the special property that each player's revenue is proportional to their share of the total cost incurred by the players (ignoring potential contast entry costs).
The classic example of such a game is the Tullock contest, introduced in two papers by Gordon Tullock~\cite{tullock1975efficient} ,\cite{tullock1980efficient}.  That model has a player's revenue equal exactly their share of the total cost, and there are no entry costs.
Chung~\cite{chung1996rent} was the first to extend the basic model by letting the total reward be an increasing function of the total cost, and study the resulting equilibria. An almost identical model has been considered by \cite{johnson2023concave} in the context of batched decentralized exchanges, establishing lower bounds on the game's price of anarchy.
A variant with fixed total reward and entry costs has been analyzed by \cite{AppelbaumKatz1986}, showing that as the entry cost increases there are less entrants in equilibrium.
Games with multiple simultaneous Tullock-like contests have also been considered \cite{grossmann2021entry}, \cite{AzmatMoller2009},
\cite{Sela2025Unsolicited}.  These all have unique elements that do not capture the game studied in section~\ref{sec:game}.

Some forms of non-amortizability have been studied in the context of proofs of ``useful'' work \cite{EPRINT:BRSV17b,EPRINT:KomWei25}. These are variants of proofs of work that also permit checking the output of a deterministic computation through an ad-hoc scheme. In contrast to prior works, we model and study  non-amortizability of \textit{existing} succinct arguments over \textit{non-deterministic} computations.

\section{Attention Mechanisms as Non-Amortizable Cryptographic Proofs}
\label{sec:mechanism}

\subsection{Intuitions behind our modeling choices}
We model attention mechanisms as a variant of (non-interactive) cryptographic proof systems~\cite{thalerbook,snarkhunt} (also called \textit{arguments} in settings where the adversary is computationally bounded).
Such schemes consist of a proving (resp. verification) algorithm through which one can certify (resp. check) whether an NP statement is true. The output of proving is a string $\pi$, a proof, which we will require to be succinct (intuitively, shorter than the NP witness showing the validity of the statement).
Proof schemes may guarantee a specific form of soundness property,  \textit{knowledge soundness}, if, informally speaking, an accepting proof guarantees that the prover actually knows the witness behind the true statement.

In our modeling, we follow cryptographic proofs as inspiration, but diverge from them in a few aspects in order to obtain something tailored to our setting:

\begin{itemize}
\item We assume each operator has a public key associated to its identity, as well as a corresponding secret key. The latter can be used to generate a proof. We require (through a property we call \textit{designated-proving}) that knowledge of a secret key corresponding to identity X is required to produce a proof associated to X.
\item Since our goal is to incentivize that operators have actually \textit{watched} execution, we make a very specific choice when modeling knowledge soundness. We model ``knowledge'' as the explicit \underline{operational} ability to retrieve the trace by interacting with a prover producing accepting proofs\footnote{More specifically, we model this through \textit{black-box} knowledge soundness, where an efficient extractor machine can rewind the adversary having access to it in contrast to \textit{non-black-box} forms of knowledge soundness where the extractor may be a machine that is merely guaranteed to exist but that may not be explicitly defined~\cite{AFRICACRYPT:CGKS23}. In this sense, our modeling reflects the difference between proofs of retrievability and proofs of data-possession, where the former enforces the ability to retrieve data that had been outsourced by a client~\cite{EPRINT:BowJueOpr08a}.}. Also, the NP relation we consider simply captures opening of a commitment (a hash).
\item One of the desired properties of our mechanism, and one that plays an important role in our game-theoretic analysis, is that an operator who registers many identities (Sybils) cannot amortize its costs across them: each additional identity must carry an irreducible cost of its own. We model this through a non-amortizability property for proving generation. The latter is novel in the realm of proving schemes and may be of independent interest.
\end{itemize}

\subsection{Arguments with Registered Prover}

All definitions  are in the random oracle model. We denote by $\roSet$ the set of all functions $\bits^* \to \bits^\lambda$; $\ro \sample \roSet$ denotes sampling of a random oracle.

\begin{definition}[Argument with Registered Provers]
    \label{def:dp-arg}
A (succinct) argument system with registered prover with message space $\mathcal{M}=\{\mathcal M_\lambda\}_\lambda$ is a tuple of efficient algorithms $\left( \setup, \KG, \mathsf{Com}, \prv, \vfy \right)$, of which only $\prv$ and $\vfy$ have access to the random oracle, with the following syntax:
\begin{itemize*}[label={}, itemjoin={}]
\item $\setup(1^\lambda) $ outputs a common reference string (CRS) pair $(\pcrs,\vcrs)$ for prover and verifier respectively\footnote{We note that, unless specified otherwise, all algorithms are randomized. We also assume without loss of generality that $\pcrs$ contains $\vcrs$. }; 
\item $\KG(\pcrs)$ outputs a secret--public pair $\left( \sk, \pk \right)$;
\item $\mathsf{Com} \left( \pcrs, \vec{a} \right)$ deterministically produces a commitment $\cm$ to a string $\vec{a} \in \mathcal M_\lambda$;
\item $\allowbreak\prv^\ro \left( \pcrs, \sk, \vec{a}, \nonce \right)$  produces a proof $\pi$ (see also verification);
\item $\vfy^\ro (\vcrs, \allowbreak \pk,$  $\cm, \nonce, \pi )$ outputs a bit for acceptance/rejection of a proof $\pi$. Intuitively verification checks that $\pi$ is a valid proof ``of knowledge of the commitment opening of $\cm$ produced by the owner of $\mathsf{SecKey}(\pk)$ in the context given by nonce $\nonce$''.
\end{itemize*}
An argument system with registered provers should satisfy completeness, binding, knowledge soundness, succinctness, designated proving and non-amortizability, defined below. 
\end{definition}

\noindent
\textit{Completeness, binding, succinctness and designated proving.}
 Completeness (\cref{def:completeness}) states that for an honestly generated CRS pair and key pair, any $\vec{a} \in \mathcal M_\lambda$ and any string $\nonce \in \bits^*$, the proof $\pi \gets \prv^\ro \left( \pcrs, \sk, \vec{a}, \nonce \right)$ satisfies $\vfy^\ro \left(  \vcrs, \pk, \Com(\pcrs, \vec{a}), \nonce, \pi \right) = 1$ with probability $1$.
 The other properties are mostly standard (\cref{def:binding,def:succinctness,def:designated-proving}).
 Binding requires that no adversary can output two strings producing the same commitment;
 Succinctness requires proofs and commitments to be short. Designated proving states that an adversary without $\sk$ cannot produce an accepting proof under $\pk$.

\begin{definition}[Completeness]
\label{def:completeness}
For all $\lambda \in \mathbb{N}$, all $\vec{a} \in \mathcal M_\lambda$ and all $\nonce \in \bits^*$:
\[ \multilinePr{
    & \vfy^\ro \bigl(  \vcrs, \pk,\\
    & \quad \Com(\pcrs,\protect\vec{a}), \nonce, \pi \bigr)=1
}{
    & \ro \sample \roSet\\
    & (\pcrs,\vcrs) \gets \setup \left( 1^\lambda \right)\\
    & \left( \sk, \pk \right) \gets \KG\left( \pcrs \right)\\
    & \pi \gets \prv^\ro \left( \pcrs, \sk, \protect\vec{a}, \nonce \right)
} = 1 \]
\end{definition}

\begin{definition}[Binding]
\label{def:binding}
For all PPT adversaries $\adv$ and all $\lambda \in \mathbb{N}$:
\[ \multilinePr{
    & \protect\vec{a} \neq \protect\vec{a}' \ \wedge\ \protect\vec{a},\protect\vec{a}' \in \mathcal M_\lambda \\
    & \quad \wedge\ \Com(\pcrs, \protect\vec{a}) = \Com(\pcrs, \protect\vec{a}')
}{
    & (\pcrs,\vcrs) \gets \setup \left( 1^\lambda \right)\\
    & (\protect\vec{a},\protect\vec{a}') \gets \adv \left( \pcrs \right)
} \leq \negl \]
\end{definition}

\begin{definition}[Succinctness]
\label{def:succinctness}
For all $\lambda \in \mathbb{N}$, all $\vec{a} \in \mathcal M_\lambda$, all $\nonce \in \bits^*$, all $\ro \in \roSet$, all $(\pcrs,\vcrs)$ in the support of $\setup(1^\lambda)$, all $\left( \sk, \pk \right)$ in the support of $\KG(\pcrs)$ and all $\pi$ in the support of $\prv^\ro(\pcrs, \sk, \vec{a}, \nonce)$:
\[ \abs{\pi} = \poly \cdot \polylog \left( \abs{\vec{a}} \right),
   \qquad
   \abs{\Com(\pcrs,\vec{a})} = \poly, \]
and $\vfy^\ro$ runs on input $(\vcrs, \pk,$ $\Com(\pcrs,\vec{a}), \nonce, \pi)$
in time
\[ \poly \cdot \left( \abs{\nonce} + \polylog \left( \abs{\vec{a}} \right) \right). \]
\end{definition}

The following property intuitively states that no adversary without a secret key can produce an accepting proof for the corresponding public key, even for a commitment whose opening it knows.
\begin{definition}[Designated Proving]
\label{def:designated-proving}
For all PPT adversaries $\adv$ and for all $\lambda \in \mathbb{N}$:
\[ \multilinePr{
    & \vfy^\ro \bigl(  \vcrs, \pk,\\
    & \quad \cm, \nonce, \pi \bigr)=1\\
    & \quad \wedge \ (\cm,\nonce) \notin \mathcal Q
}{
    & \ro \sample \roSet\\
    & (\pcrs,\vcrs) \gets \setup \left( 1^\lambda \right)\\
    & \left( \sk, \pk \right) \gets \KG\left( \pcrs \right)\\
    & \left( \cm,\nonce,\pi \right) \gets\adv^{\ro, \prv^\ro_\sk} \left( \pcrs,\pk \right)
} \leq \negl \]
where $\prv^\ro_\sk \coloneqq \prv^\ro \left( \pcrs, \sk, \cdot, \cdot \right)$ and $\mathcal Q$ is the set of pairs $(\Com(\pcrs,\vec{a}'),\nonce')$ over the queries $(\vec{a}',\nonce')$ that $\adv$ makes to $\prv^\ro_\sk$.
\end{definition}

\begin{definition}[Knowledge Soundness]
\label{def:ks}
There exists an efficient extractor $\ext$  such that for all PPT adversaries $\adv$ and  all $\lambda \in \mathbb{N}$, this quantity is negligible:
\[ \multilinePr{
    & \vfy^\ro \left(  \vcrs, \pk, \cm, \nonce, \pi \right)=1  \\
    & \quad \wedge \ \cm \neq \Com \left( \pcrs, \protect\vec{a} \right)
}{
    & \ro \sample \roSet; \ (\pcrs,\vcrs) \gets \setup \left( 1^\lambda \right)\\
    & \left( \sk, \pk \right) \gets \KG\left( \pcrs \right)\\
    & r \sample \bits^{\poly} \\
    & \left( \cm,\nonce, \pi \right) \gets \adv^\ro \left( \pcrs,\sk,\pk ; r \right)\\
    & \protect\vec{a} \gets \ext^{\adv,\ro} \left( \pcrs,\sk,\pk,\cm, \nonce, r \right)
} \]
where $r$ is the random tape of $\adv$. $\ext$ may run $\adv$ on inputs and random tapes of its choice and answer each random oracle query of $\adv$ with a value of its choice.
\end{definition}

The following definition informally states that an adversary can produce $n$ valid proofs (for $n$ different keys) only by running in time roughly  $n \cdot \Delta$,
where, intuitively, $\Delta$ is some non-vanishing fraction of the time required to generate one individual proof. 
\begin{definition}[Non-Amortizability]
    \label{def:non-amort}
    Let $\Delta$ be a  function. We say that an argument with registered prover   satisfies non-amortizability with overhead $\Delta$ over a distribution $\xwDist$ on $\mathcal M_\lambda$ if for all $n = \poly$ and for all $\adv$ such that $\mathsf{Time}(\adv)/ n(\lambda) =  o\left(\Delta(\lambda)\right)$ for all $\lambda\in \mathbb{N}$, this quantity is negligible:

   \[ \multilinePr{
    & \bigwedge_j \vfy^\ro \Big(  \vcrs, \pk_j,\\ & \ \ \ \Com(\pcrs,\protect\vec{a}), \nonce, \pi_j \Big)=1\!\!\!
}{
  \!\!  & \ro \sample \roSet;\    (\pcrs,\vcrs) \gets \setup \left( 1^\lambda \right)\\
    \!\!& \left( \sk_j, \pk_j \right) \gets \KG\left( \pcrs \right), \  j \in [n]; \ 
     \protect\vec{a} \sample \xwDist \\
    \!& \nonce, \left( \pi_j \right)^n_{j=1} \!\gets
     \adv^\ro\!\! \left( \pcrs,\left( \sk_j,\pk_j \right)^n_{j=1}, \protect\vec{a}\right)
} \]

\end{definition}

\subsection{The TRACE Construction as an Argument with Registered Prover}
\label{sec:ipa-construction}

The claim procedure of \cref{alg:both} is a variant of that of TRACE~\cite{trace}.
We recast it here as an argument with registered prover, $\Pi_{\mathsf{IPA}}$.
We instantiate it with a Pedersen vector commitment under a structured key and
with the publicly verifiable inner-product argument of Daza, R\`afols and
Zacharakis~\cite{drz}, recalled in \cref{fig:pv-ipa}, rather than with the
transparent argument of Bulletproofs (\cref{fig:bp-ipa}).
The inner-product argument is used as a black box.
The final construction also signs each proof with a signature key registered
alongside $\pk$, which is what gives designated proving.
The protocol is prior work, in a slightly modified form\footnote{The original version of TRACE  performs an inner product (roughly) of the form $\langle E, s\vec{1}\rangle$ for a secret $s$. For this to work, we would need an explicit zero-knowledge guarantee(and potentially simulation-extractability, a stronger property)  on the inner product argument. Instead we can obtain the same effect without zero-knowledge via a signature and by making sure that the random oracle challenges use the public key of the node.}.
What the recasting supplies is \cref{def:non-amort} and the bound of
\cref{thm:ipa-nonamort}, neither of which the original states.

\paragraph{Leveraging Inner-Product Arguments for attention.}
Write $N=2^r$ for the padded trace length and $[\mathbf r]_1,[\mathbf s]_1\in\mathbb G_1^N$
for the two structured keys of \cref{fig:pv-ipa}.
Messages are traces $E\in\mathcal M_\lambda=\mathbb F_p^N$, and
$\Com(\pcrs,E)=[E^{\top}\mathbf r]_1$ is their Pedersen commitment $C_E$, a
consensus output of the state transition function rather than a value the
claimant chooses.
The nonce is the block hash.
We run the inner-product argument on the trace and the all-ones vector, whose
commitment $[S]_1=[\mathbf 1^{\top}\mathbf s]_1$ is published in the reference
string.
The public key is the verification key of a signature scheme.
A proof carries the sum $v=\langle E,\mathbf 1\rangle$, an inner-product proof,
and a signature on both, which guarantees designated proving.
The inner-product argument certifies the relation
\[
    \rel_{\mathsf{claim}}
    =\left\{
      \bigl((\cm,v),\,E\bigr)
      \;\middle|\;
      \cm=\Com(\pcrs,E),\;
      v=\langle E,\mathbf 1\rangle
    \right\}.
\]

The registered key enters the inner-product argument in exactly one place: the
seed $\theta$ of the Fiat--Shamir oracle handed to $\mathsf{P}$, which hashes
$\pk$ together with the commitment, $v$, and the nonce.
This is essential.
It is the only input that makes one identity's inner-product proof differ from
another's, and without it an operator could compute a single proof and sign it
under every identity it controls.
For the same reason the eligibility test of \cref{alg:both} hashes the
inner-product proof and not the signature, so that a randomized signature
cannot be resampled until it wins.
The main reason why a Bulletproofs-style inner-product argument is useful here is its natural non-amortizability.
\Cref{app:construction} gives the construction, with proof sketches of
completeness, binding, knowledge soundness, succinctness, and designated
proving.

\begin{theorem}
\label{thm:construction-props}
Let $(\mathsf{Gen},\mathsf P,\mathsf V)$ be a labelled inner-product argument
whose key generator satisfies Find-Rep (\cref{def:findrep}) and which is
perfectly complete, black-box knowledge sound, and succinct, with proofs and
verification logarithmic in $N$ (\cref{def:ip-props}).
Let $\Sigma$ be a correct signature scheme that is existentially unforgeable
under chosen-message attack (\cref{def:sig-props}).
Then $\Pi_{\mathsf{IPA}}$ (\cref{fig:construction}) is an argument with
registered provers that satisfies completeness, binding, knowledge soundness,
succinctness, and designated proving
(\cref{def:completeness,def:binding,def:ks,def:succinctness,def:designated-proving}).
\end{theorem}
\noindent
Proof sketches are in \cref{app:construction-properties}.
The argument of \cref{fig:pv-ipa} meets these requirements~\cite[Thm.~4]{drz}.

\paragraph{Non-amortizability.} In order to analyze non-amortizability, we need to provide a  new analysis (this type of property was never studied in proof systems to the best of our knowledge). 
A different succinct argument for the same relation may well admit
a prover that amortizes across keys.
\Cref{thm:ipa-nonamort} is accordingly proved for this instantiation, by
counting group operations round by round in the generic bilinear group model.
We carry this analysis out in \Cref{sec:crypto-analysis}.
One consequence is visible already here: the first round of the argument
depends only on the trace and the reference string, so every identity of an
operator shares it, and the bound comes entirely from the rounds that follow the
first challenge, which is the first point where the identity's key enters.

\begin{theorem}[Non-amortizability of $\Pi_{\mathsf{IPA}}$]
\label{thm:ipa-nonamort}
Let $N=2^r$ and $\log p$ be polynomially bounded in $\lambda$ with
$\log p=\omega(\log\lambda)$, and set
\[
    \Delta(\lambda)=\frac{(r-1)\log p}{\log\lambda}.
\]
Then, in the generic bilinear group model with a random oracle,
$\Pi_{\mathsf{IPA}}$ satisfies non-amortizability with overhead $\Delta$ (\cref{def:non-amort}), over every distribution $\xwDist$.
\end{theorem}

We describe a PRF-based variant with stronger non-amortizability guarantees in \cref{sec:prf-cons}. It makes the
per-identity secret vector fresh by construction, but requires a general-purpose
zero-knowledge proof and has a more expensive prover.

\begin{table}[t]
\centering
\caption{Comparing Proof of Diligence (PoD)~\cite{pod} and TRACE~\cite{trace}. $N$ is the padded trace length and $\epsilon \in (0,1)$ is arbitrary.}
\label{tab:properties}
{\setlength{\tabcolsep}{3pt}
\begin{tabular}{@{}lcccc@{}}
\toprule
 & Succinct
 & \shortstack{Designated proving}
 & \shortstack{Knowledge of trace}
 & \shortstack{Non-amortizability} \\
\midrule
PoD &  yes & yes & no & $\Delta=O(1)$ \\
TRACE &  yes & yes & yes
  & $\Delta = \Omega \left( \log N \cdot \lambda^{1-\epsilon} \right)$ \\
\bottomrule
\end{tabular}
}
\end{table}

\section{The Game}
\label{sec:game}
In this section, node operators are referred to as players. Let $g_v,g_a:\mathbb{N}\rightarrow \mathbb{R}$ denote total reward functions of the overall number of validator and attention sybils, respectively, which are given to validators and attention players. We assume that $g_v,g_a$ are both concave, twice differentiable and satisfy $g_v(0)=g_a(0) = 0$. 

$c$ denotes the entry cost per player, $d_v$ the cost per sybil in the validation game and $d_a$ the cost per sybil in the attention game. We assume that $d_v>d_a$. This introduces (mathematical) asymmetry between these two games. 

The strategy set of each player is to choose which game it joins and how many sybils it deploys in each game. 

Player $i$'s utility from deploying $k_v^i$ sybils in the validation game and $k_a^i$ sybils in the attention game, when there are a total of $K_v$ sybils deployed by all players in the validation game and $K_a$ sybils deployed by all players in the attention game is defined as: 

\begin{equation}\label{utility_validation}
U^i(k_v^i,k_a^i) = \frac{k_v^i}{K_v}g_v(K_v)-d_vk_v^i+\frac{k_a^i}{K_a}g_a(K_a)-d_ak_a^i - c\mathbbm{1}_{\{k_v^i > 0 \text { or } k_a^i >0\}}.    
\end{equation}

For a game $j\in \{validation, attention\},$ we define the average reward per sybil in game $j$ as 

\begin{equation*}
    r_j(n):=\frac{g_j(n)}{n}, n\geq 1.
\end{equation*}
Concavity of $g_j$ guarantees that $r_j(n)$ is non-increasing. 

We focus on a specific type of equilibrium of this (meta) game. For given parameters $V$ and $A$ such that $V<A$, $A$ players participate in the attention and $V$ \emph{of those} $A$ players also participate in the validation game. The remaining $A-V$ players participate in the attention game only. All other potential players do not participate in any game, i.e., they stay inactive.

If a player deploys $k$ sybils while others deploy $m$ in game $j$, its payoff in that game ignoring the entry cost is: 

\begin{equation}\label{payoff_fixed_game}
    \frac{k}{k+m}g_j(m+k)-d_jk=k(r_j(m+k)-d_j).
\end{equation}

Suppose we target $s$ sybils per player in game $j$, when there are $n$ active players in game $j$. Define the payoff of a single active player who instead deploys $q$, while all $n-1$ others deploy $s$:

\begin{equation}\label{deviation_to_q}
    D_j(q|n,s):=q(r_j((n-1)s+q)-d_j).
\end{equation}
For a new player joining this game with $q$ sybils the payoff is defined as:

\begin{equation}\label{new_entrant}
    E_j(q|n,s) := q(r_j(ns+q)-d_j).     
\end{equation}
In the equilibrium with $s_v$ validator sybils per player, best-response guarantees: 

\begin{equation}\label{eq_condition_val1}
    D_v(s_v|V,s_v) \geq D_v(q|V,s_v)\ ,\ \forall q\neq s
\end{equation}
The other equilibrium condition is a no-entry condition for inactive players: 

\begin{equation}\label{no_entry_validation}
    \sup_{q\geq 1}E_v(q|V,s_v) \leq 0.
\end{equation}
The best-response, individual rationality and no-entry conditions for the attention game are: 

\begin{align}
    D_a(s_a|A,s_a) &\geq D_a(q|A,s_a), \qquad \forall q\neq s_a
    \label{best_response_attention}\\
    D_a(s_a|A,s_a) &\geq c
    \label{individual_rationality_attention}\\
    \sup_{q\geq 0} E_a(q|A,s_a) &\leq c
    \label{no_entry_attention}
\end{align}

A particular outcome of interest is when each player deploys at most a single sybil in both attention and validation games.  
We now study the conditions under which this outcome is attained.

\paragraph{Validation game.}

Since the average reward function is non-increasing in the number of sybils, the validation condition reduces in this case to 

\begin{equation}
\label{eq:validation_game_condition}
    \frac{g_v(V)}{V}\geq d_v \geq \frac{g_v(V+1)}{V+1}.
\end{equation}
Indeed, this ``single-crossing" condition handles all deviations from the equilibrium:
the active validator earns $\frac{g_v(V)}{V}-d_v\geq0$, satisfying individual rationality. 
A non-validator deviating to deploying $q\geq 1$ sybils earns $q(\frac{g_v(V+q)}{V+q}-d_v) \leq 0$, where the inequality is implied from the decreasing average rewards established above. Similarly, an active validator is worse off deploying $q\geq 2$ sybils instead of one. 

\paragraph{Attention game.}
We now analyze the attention game equilibrium conditions in this setting. Let $p_a:=\frac{g_a(A)}{A}-d_a$ denote the attention game profit per sybil (ignoring entry cost). The one sybil conditions are: 

\begin{equation*}\label{one_Sybil_attention}
    p_a\geq c, \quad \ 
q(\frac{g_a(A-1+q)}{A-1+q}-d_a)\leq p_a \ \text{for } q\geq 2
\quad\     q(\frac{g_a(A+q)}{A+q}-d_a)\leq c,\ \text{for } q\geq 2
\end{equation*}

Next, we identify the precise condition under which a one-sybil per player equilibrium exists (we call this a ``good equilibrium''):
\begin{equation}\label{condition_for_good_equilibrium}
c\leq \frac{2d_a}{A-1}.
\end{equation}
The proof of the following result is deferred to Appendix~\ref{app:prop_necessity_of_condition}.

\begin{proposition}\label{prop:necessity_of_condition}
    Condition~\eqref{condition_for_good_equilibrium} is necessary for a good equilibrium to exist for any non-decreasing reward function $g_a$. 
\end{proposition}

For the sufficiency, consider the simplest reward function, namely fixed pools of rewards:  $g_j(0)=0, g_j(n)=R_j$ for a game $j$, and any $n>0$.

In this case, the equilibrium condition for the validation game translates into 
\begin{equation}
\label{interval_for_validation_reward}
    d_v V \leq R_v \leq d_v(V+1).
\end{equation} 

The condition for the attention game is: 

\begin{equation}\label{interval_for_attention_reward}
    A(d_a+c) \leq R_a \leq \min\{ \frac{d_a A(A+1)}{A-1}, \min_{q\geq 1}(A+q)(d_a+\frac{c}{q})\}
\end{equation}

Note that $(A+q)(d_a+\frac{c}{q})-A(d_a+c)=qd_a+\frac{Ac}{q}-c(A-1)$. For $q=1$, this difference equals to $d_a+c>0$. For $q\geq 2$, $qd_a\geq 2d_a\geq c(A-1)$, exactly the condition from the proposition claim, i.e., $A(d_a+c)<\min_{q\geq 1}(A+q)(d_a+\frac{c}{q})$.

For any $A$ which satisfies Condition~\eqref{condition_for_good_equilibrium}
and for any $V \leq A$ we can find appropriate constant reward functions $R_V, R_A$ to satisfy the equilibrium conditions~\eqref{interval_for_validation_reward} and \eqref{interval_for_attention_reward}.

The same argument extends to $s_a>1$. If $M$ players each deploy $s_a$ sybils,
so that $A=M s_a$, the necessary condition has the form
\begin{equation*}
  c < \frac{d_a s_a(s_a+1)}{A-s_a}.
\end{equation*}
The full statement and proof are in \cref{app:multi-sybil}.

\paragraph{Non-monotonic attention reward function}
The discussion above shows that whenever Condition~\eqref{condition_for_good_equilibrium} imposes an upper bound on $A$. This is too restrictive for the values $d_A$ and $c$ in practice.
If we allow reward functions that may decrease from some point then the equilibrium conditions can be relaxed.

For example, choose $p_a>c$, $\beta > \max\{\frac{p_a}{2},p_a-c\}$ and a target amount of attention nodes $A$, and consider the smooth concave function
\begin{equation}
\label{eq:non_monotonic_reward}
g_a(x)=(d_a+p_a+\beta A)x-\beta x^2.
\end{equation}
Then, at point $A$, we have: 

\begin{equation}
    \frac{g_a(A)}{A} - d_a=p_a>c.
\end{equation}

If an active player uses $q$ sybils while all others use one sybil each, its payoff is $q(p_a-\beta(q-1))$, which is strictly below $p_a$ for any $q\geq 2$.

An inactive player entering the participation games with $q$ sybils obtains $q(p_a-\beta q)$, which is below $c$ for every $q\geq 1$. Therefore, the good equilibrium conditions are met, $A$ players deploy exactly one sybil, for an arbitrary $c$. The only ``problem" here is that $g_a(x)$ becomes negative at some point, but we can truncate it.~\footnote{As is suggested by the following Ethereum improvement proposal (EIP): https://eips.ethereum.org/EIPS/eip-8363. The truncated version looks qualitatively exactly the same as the EIP-8363.} 

Decreasing rewards support the desired equilibrium even when $d_a$ is small.
  Section~\ref{sec:tradeoff} shows how $d_a$ controls disruption costs and the
  headroom before collapse.

\section{Worked Example: BoLD and TRACE}
\label{sec:example}

In this section we instantiate the game of \cref{sec:game} with the costs an
operator faces on Arbitrum One, and report how many attention nodes the design
sustains.
Assertions are posted once an hour, so we take one hour as the period and
express every cost per hour.

Posting or challenging an assertion requires stake of about \$9.8M, so at an
annual rate of 5\% the hourly opportunity cost of validation is
\begin{equation}
\label{eq:example_d_v}
  d_v \;\approx\; 9.8 \times 10^{6} \cdot \frac{0.05}{8760} \;=\; \$56,
\end{equation}
where $8760$ is the number of hours in a year.
A managed provider charges about a dollar an hour for an Arbitrum full
node~\cite{chainstack}, so $c = \$1.00$ per hour.
An identity proves on every block, Arbitrum produces one every $0.25$ seconds,
and one TRACE proof costs about $4N$ points of multi-scalar multiplication at
trace length $N$.
At $N = 2^{20}$ that is $2.9$ machine-hours of proving per hour, which priced
against cloud hardware gives $d_a = \$0.91$ per hour.
\Cref{app:costing} carries the hardware, the benchmarks, and the arithmetic.

\subsection{Execution Diversity for Non-Decreasing Reward Functions.}
Condition~\eqref{condition_for_good_equilibrium} rearranges to
\begin{equation}
\label{eq:bound}
  A \;\leq\; 1 + \frac{2 d_a}{c},
\end{equation}
where $A$ is the number of attention players and the bound is the largest
population for which one sybil per player is an equilibrium, when the attention
reward function is non-decreasing.
Substituting gives $A \leq 2$ at $N = 2^{20}$.
Trace length is the only parameter of the mechanism that moves $d_a$, and
\cref{app:costing} sweeps it together with what the resulting proving costs.
Lowering the rate at which claims win does not raise the bound, because every
identity proves on every block whether or not it wins.
Under a non-decreasing reward an operator gains from extra identities, so the
only way to stop it is to make each identity expensive.

\subsection{Execution Diversity with Non-Monotonic Attention Reward Functions}
\label{subsec:non_monotonic}

As explained above, relaxing the monotonicity restriction evades Condition~\ref{condition_for_good_equilibrium}.

For example,
suppose the rollup operator is content with $A = 30$ attention node operators, out of which $V=5$ are also validators.
In this case, using~\ref{eq:example_d_v} we can set the constant hourly validator reward function
\[
R_v = 5.5 \cdot d_v \approx \$308
\]
and this would satisfy the equilibrium condition for validators~\ref{interval_for_validation_reward}.

As for the attention reward function, we can choose any parameter $p_a > c$ such as $p_a = 2.09$, a parameter $\beta > \max\{\frac{p_a}{2},p_a-c\}$ such as $\beta = 2$.
Plugging into~\ref{eq:non_monotonic_reward},we end up with the attention reward function
\begin{equation*}
    \label{eq:non_monotonic_reward_ex}
    g_a(x) = x(63 - 2x)
\end{equation*}

The results above show that under these reward functions, $A=30$ attention nodes out of which $V=5$ are validators, where each operator deploys a single sybil in each role they play, is indeed an equilibrium outcome.

\subsection{Tradeoff between the reward curves}
\label{sec:tradeoff}
The conditions above ask whether a single operator wants to
participate in attention, and establish an equilibrium. It is not the only equilibrium of the game.
Consider a pool of $\$90$ and a target of $30$ identities. With one identity per
operator, each nets $1.09$. But fifteen operators with two identities each net
$3.18$. Dropping one lowers that payoff to $3.09$, while adding one or entering
as a new operator loses money. The remaining checks are in
\cref{app:reward-tradeoff}.

We now consider two further questions: i) how the population reaches an outcome,
and ii) what an attacker can do to an established one. The answers depend on
$\beta$, the slope of the decreasing reward, and on $d_a$.
Dropping an identity raises every remaining identity's payment by $\beta$.
Adding a second identity is unattractive when $\beta>p/2$, while dropping from
two identities is unattractive when $\beta<p$. More generally, any $\beta>p$
makes operators with more than one identity unstable. Here $\beta=2$ and
$p=2.09$, so the target is stable to these deviations.

The reward pool reaches zero at $A+(d_a+p)/\beta$, so a higher $\beta$ brings
collapse closer to $A$. A constant reward limits how many operators the mechanism
pays, while a decreasing one limits how far identities can overshoot $A$.
The bound on that overshoot and its derivation are in
\cref{app:reward-tradeoff}. In both cases $d_a$ is the cost of an identity, so a
small $d_a$ lowers the cost of crossing the target.

\section{A Stronger Attention Construction and Its Limitations}
\label{sec:variant-limitation}

The cost in \cref{sec:example} is the work of an honest prover, not a lower
bound for the main construction.
\Cref{sec:prf-cons} gives a stronger non-amortizability guarantee. It has a more
expensive prover, and it is not the construction priced in
\cref{sec:example}.
For the main IPA construction, \cref{thm:ipa-nonamort} gives only a logarithmic
overhead in the trace length. The sharing construction in
\cref{app:amortization} reduces the cost to $N^{0.936}$ per identity, where
$N$ is the padded trace length.
The gap comes from the vector the main construction folds the trace against.
It folds the trace against the all-ones vector, which is the same for every
identity, so the vector side of each IPA round adds no per-identity work.

Deriving $\mathbf b$ from the registered key removes that shared computation.
But key derivation alone does not restore a recurring cost proportional to the
trace. \Cref{app:keyed} describes this preprocessing variant.
Because $\mathbf b$ is fixed at registration, an operator can precompute the
part of each round that does not depend on the block. The remaining secret-side
work is $\tilde O(N^{2/3})$ per block.
The work is not removed; it moves to registration. Each identity pays its own
one-time $\tilde O(N^{4/3})$ charge, which an operator cannot share across
identities.
To express that one-time cost as an hourly cost, we would need to assume how
long an identity remains registered and measure the registration proof. We do
neither here.
It remains open whether making the folding order depend on the block would
restore a recurring per-block cost.

The PRF construction in \cref{sec:prf-cons} takes a different approach: it
derives $\mathbf b$ from the block context and proves the derivation in each
claim. This is the construction that gives the stronger non-amortizability
guarantee.

None of this changes \cref{sec:example}. That section prices the work of an
honest prover on every block, measured on hardware rather than derived from a
lower bound.
 
The conditions above show that the target population is an
equilibrium. They do not show that it is unique or that entry and exit will
reach it.
\Cref{sec:tradeoff} gives a second equilibrium with the same population
numbers.
The model assumes that operators act alone. It does not analyze collusion.
The baseline costs in \cref{sec:example} use Arbitrum One at one trace length.
We measure multi-scalar multiplication on one machine and extrapolate the two
largest trace lengths.
The bound in \cref{eq:bound} applies whenever an operator shares one node cost
across its identities and each identity has its own additional cost.
If a mechanism charges an identity only a deposit, then $d_a$ is fixed by that
deposit. The bound therefore does not move with parameters that leave the
deposit unchanged.

\section{Conclusion}
\label{sec:conclusion}

We studied how an attention reward translates into execution diversity when one
operator can register many identities. Modeling validation and attention as
coupled contests shows that, under a non-decreasing reward, one identity per
operator requires $A \leq 1+2d_a/c$, so the per-identity cost $d_a$ is the lever
a deployer controls.
Arguments with registered provers and their non-amortizability make that cost
precise.
Proof of Diligence has $\Delta=O(1)$ and certifies trace roots rather than the
trace, while our variant of TRACE is knowledge sound and has an overhead
logarithmic in the trace length.
Measured costs on Arbitrum One support $A\leq2$ at trace length $2^{20}$. A
decreasing reward can support a larger target without raising $d_a$, but it
moves the constraint to the reward curve.

For TRACE, the proven overhead is logarithmic in the trace length, the best
sharing we know of costs $N^{0.936}$ per identity, and the honest prover pays
$\Theta(N)$.
Closing this gap requires lower bounds for multi-scalar multiplication with
structured exponents and preprocessing.
The PRF construction of \cref{sec:prf-cons} makes the overhead linear, at the
price of a more expensive prover and an analysis that models the PRF
as a random oracle.
More broadly, non-amortizability runs opposite to the batching and aggregation
that most work on succinct arguments pursues, and it deserves a systematic
study: which proof systems are non-amortizable and with what overhead, whether
generic transformations can raise the overhead of an existing argument, and
whether such bounds hold outside idealized models such as the generic group
model.
The notion may also apply beyond attention, to any mechanism that pays or
rate-limits registered identities and needs each one to do its own work, such as
proof-of-stake committees whose members are expected to hold or check data
independently.

\bibliographystyle{alpha}
\newcommand{\etalchar}[1]{$^{#1}$}

\appendix
\etocdepthtag.toc{appendix}
\begingroup
\etocsettagdepth{mainmatter}{none}
\etocsettagdepth{appendix}{subsection}
\etocsettocdepth{subsection}
\etocsetlevel{title}{0}
\etocsetlevel{author}{0}
\def\authcount#1{}
\etocsettocstyle{\section*{Contents of the Appendix}}{}
\tableofcontents
\endgroup

\section{Non-Amortizability of the Inner-Product Arguments}
\label{sec:crypto-analysis}

\subsection{Preliminary definitions and properties for non-amortizability}

\begin{definition}[Hill Computational min-entropy]
Let $X = \{X_\lambda\}_{\lambda \in \mathbb{N}}$ be a distribution
ensemble over domains $\mathcal{X}_\lambda$.
We say that $X$ has at least $\kappa(\lambda)$ bits of computational
entropy if there exists a distribution ensemble
$Y = \{Y_\lambda\}_{\lambda \in \mathbb{N}}$ over the same domains such that
\[
    X \approx_c Y
    \qquad\text{and}\qquad
    H_\infty(Y_\lambda)
    := -\log_2 \max_{y \in \mathcal{X}_\lambda}
                  \Pr[Y_\lambda = y]
    \ge \kappa(\lambda).
\]
Here, $X \approx_c Y$ means that, for every PPT distinguisher $D$,
there exists a negligible function $\nu$ such that
\[
    \left|
        \Pr[D(1^\lambda,X_\lambda)=1]
        -
        \Pr[D(1^\lambda,Y_\lambda)=1]
    \right|
    \le \nu(\lambda).
\]
\end{definition}

\paragraph{Generic group model.}
Let $N=N(\lambda)$ and $p=p(\lambda)$ be efficiently computable,
where $N$ and $\log p$ are polynomially bounded, $p$ is prime,
and $\log p=\Omega(\lambda)$.
Let $G_\lambda$ be an additive cyclic group of order $p$,
and set $\ell=2\lceil\log_2 p\rceil$.
A generic encoding is a uniformly random injection
$\sigma:G_\lambda\to\{0,1\}^{\ell}$.
The oracle $\mathcal O_\sigma$, on previously issued handles
$\alpha=\sigma(x)$ and $\beta=\sigma(y)$ and a sign
$s\in\{-1,1\}$, returns $\sigma(x+sy)$.
An invalid query returns $\bot$.
The adversary may compare handles and perform arbitrary local
computation, but its output must be an issued handle.
Each oracle query costs at least one time step, and all computation,
including preprocessing, is charged to the adversary's running time.
No precomputed group elements or auxiliary advice about the sampled
bases are supplied.

\begin{definition}[MSM experiment]
For a distribution ensemble
$\mathcal D=\{\mathcal D_\lambda\}_\lambda$ over $\mathbb F_p^N$,
the experiment $\mathsf{MSM}_{\mathcal A,\mathcal D}(\lambda)$
proceeds as follows:
\begin{enumerate}
    \item Sample $\vec v\leftarrow\mathcal D_\lambda$.
    Independently sample $g_1,\ldots,g_N$ uniformly from
    $G_\lambda\setminus\{0\}$ and sample a generic encoding $\sigma$.

    \item Initialize the issued handles with
    $\sigma(0),\sigma(g_1),\ldots,\sigma(g_N)$ and run
    \[
        z\leftarrow
        \mathcal A^{\mathcal O_\sigma}
        \bigl(1^\lambda,p,\vec v,
              \sigma(0),\sigma(g_1),\ldots,\sigma(g_N)\bigr).
    \]

    \item Output $1$ if $z$ is an issued handle and
    \[
        z=\sigma\!\left(\sum_{i=1}^N v_i g_i\right);
    \]
    otherwise output $0$.
\end{enumerate}
All sampling and the adversary's random coins are independent.
\end{definition}

\begin{theorem}[Entropy lower bound for MSM]
\label{thm:msm-entropy-lower-bound}
Suppose $V_\lambda\leftarrow\mathcal D_\lambda$ has at least
$\kappa(\lambda)$ bits of HILL computational min-entropy.
Then every PPT generic adversary $\mathcal A$ with worst-case
running time
\[
    T(\lambda)=o\!\left(\frac{\kappa(\lambda)}{\log\lambda}\right)
\]
satisfies
\[
    \Pr\!\left[
        \mathsf{MSM}_{\mathcal A,\mathcal D}(\lambda)=1
    \right]
    =\operatorname{negl}(\lambda).
\]
The probability is over the vector, bases, encoding,
and adversary's coins.
\end{theorem}

\begin{proof}

At the high-level, the proof leverages a counting argument, showing that the ``short'' computations (those running in $o\!\left(\frac{\kappa(\lambda)}{\log\lambda}\right)$ steps) cannot possibly account for the entropy of the final output, except with negligible probability. The rest of the proof takes care of the technicalities related to generic groups and computational (as opposed to statistical) entropy.

Fix such an adversary. We take $T\ge1$ and set $L=N+T+1$.
Since $\mathcal A$ is PPT, both $T$ and $L$ are polynomially bounded.

\emph{Symbolic simulation.}
Simulate $\mathcal A$ by associating the initial handles with
the formal coefficient vectors
$0,e_1,\ldots,e_N\in\mathbb F_p^N$.
An oracle query on vectors $a,b$ with sign $s$ produces $a+sb$.
Equal coefficient vectors receive the same handle; distinct
vectors receive distinct uniformly random labels.
This simulation is polynomial time: it stores polynomially many
vectors of polynomial length.
Sampling fresh labels can be implemented in strict polynomial
time up to negligible statistical error, which we absorb below.
Let $C_{\mathcal A}(v)$ denote the output's formal coefficient
vector on input $v$, with an invalid output treated as failure.

Fix a generator $h$ of $G_\lambda$ and write $g_i=a_i h$ for
independent uniform $a_i\in\mathbb F_p^\times$.
For distinct coefficient vectors $c,d$, the probability that
\[
    \sum_i c_i g_i=\sum_i d_i g_i
\]
is at most $1/(p-1)$: after fixing all but one $a_i$ with
$c_i\ne d_i$, the equality determines at most one value for
the remaining $a_i$.
Conditioned on $v$ and the symbolic transcript, the formal
vectors are independent of these hidden scalars.
There are at most $L$ formal vectors in the simulation;
including $v$ as an additional target vector gives at most $L+1$.
Coupling the real and symbolic executions and taking a union
bound over pairs of distinct vectors therefore gives
\begin{align*}
    \Pr[\mathsf{MSM}_{\mathcal A,\mathcal D}(\lambda)=1]
    &\le
    \Pr[C_{\mathcal A}(V_\lambda)=V_\lambda]
    +O(L^2/p)+\operatorname{negl}(\lambda)\\
    &=
    \Pr[C_{\mathcal A}(V_\lambda)=V_\lambda]
    +\operatorname{negl}(\lambda).
\end{align*}
Here the coupling can fail only upon an accidental equality
between distinct formal vectors; a correct output with the
wrong formal coefficient vector is also one of the counted
equalities.

\emph{Transfer to statistical min-entropy.}
The test that symbolically simulates $\mathcal A$ on input $v$
and accepts exactly when $C_{\mathcal A}(v)=v$ is PPT.
Computational indistinguishability thus implies
\[
    \Pr[C_{\mathcal A}(V_\lambda)=V_\lambda]
    \le
    \Pr[C_{\mathcal A}(W_\lambda)=W_\lambda]
    +\operatorname{negl}(\lambda).
\]

\emph{Count short computations.}
Let $\mathcal C_{N,T}\subseteq\mathbb F_p^N$ contain all vectors
obtainable from $0,e_1,\ldots,e_N$ using at most $T$ additions
or subtractions.
Each instruction chooses a sign and two indices among at most
$L$ available vectors, and the output chooses one of at most
$L$ vectors.
Padding shorter computations with additions of zero yields
\[
    |\mathcal C_{N,T}|\le L(2L^2)^T.
\]
Consequently,
\begin{align*}
    \Pr[C_{\mathcal A}(W_\lambda)=W_\lambda]
    &\le \Pr[W_\lambda\in\mathcal C_{N,T}]\\
    &\le |\mathcal C_{N,T}|\,2^{-\kappa}\\
    &\le 2^{-\kappa+O((T+1)\log\lambda)}.
\end{align*}
Combining the preceding bounds gives
\[
    \Pr[\mathsf{MSM}_{\mathcal A,\mathcal D}(\lambda)=1]
    \le
    2^{-\kappa+O((T+1)\log\lambda)}
    +\operatorname{negl}(\lambda).
\]
Finally, $T\ge1$ and $T=o(\kappa/\log\lambda)$ imply
\[
    (T+1)\log\lambda=o(\kappa)
    \qquad\text{and}\qquad
    \kappa=\omega(\log\lambda).
\]
Hence the displayed bound is
\[
    2^{-(1-o(1))\kappa}+\operatorname{negl}(\lambda)
    =\operatorname{negl}(\lambda),
\]
as claimed.
\end{proof}

\begin{figure}[tp]
\centering
\fbox{\begin{minipage}{\dimexpr\linewidth-2\fboxsep-2\fboxrule\relax}
\small

\textbf{Inner Product Relation.}
Given $N=2^k$, $\mathbf g,\mathbf h\in\mathbb G^N$, and
$u\in\mathbb G$, define
\[
\mathcal R_{\mathrm{BP\text{-}IPA}}
=
\left\{
 \bigl((N,\mathbf g,\mathbf h,u),P,(\mathbf a,\mathbf b)\bigr)
 \;\middle|\;
 P=\mathbf g^{\mathbf a}\mathbf h^{\mathbf b}
   u^{\langle\mathbf a,\mathbf b\rangle}
\right\},
\]
where $\mathbf a,\mathbf b\in\mathbb F_p^N$.
For an even-length vector $\mathbf v$, write
$\mathbf v_L,\mathbf v_R$ for its first and second halves.
The operation $\circ$ denotes componentwise multiplication.

\medskip
\textbf{Interaction Phase.}
Set
\[
\begin{gathered}
N_0\gets N,\qquad
\mathbf g^{(0)}\gets\mathbf g,\qquad
\mathbf h^{(0)}\gets\mathbf h,\qquad
P^{(0)}\gets P,\\
\mathbf a^{(0)}\gets\mathbf a,\qquad
\mathbf b^{(0)}\gets\mathbf b.
\end{gathered}
\]
For $i=1,\ldots,k$:
\begin{enumerate}[
    label=\arabic*., 
    leftmargin=*,
    itemsep=0.4em,
    topsep=0.3em
]
\item
$\mathcal P$ sets $N_i\gets N_{i-1}/2$ and computes
\[
c_L=
\left\langle
  \mathbf a_L^{(i-1)},\mathbf b_R^{(i-1)}
\right\rangle,
\qquad
c_R=
\left\langle
  \mathbf a_R^{(i-1)},\mathbf b_L^{(i-1)}
\right\rangle,
\]
and
\[
\begin{aligned}
L_i
&=
\bigl(\mathbf g_R^{(i-1)}\bigr)^{\mathbf a_L^{(i-1)}}
\cdot
\bigl(\mathbf h_L^{(i-1)}\bigr)^{\mathbf b_R^{(i-1)}}
\cdot u^{c_L},
\\
R_i
&=
\bigl(\mathbf g_L^{(i-1)}\bigr)^{\mathbf a_R^{(i-1)}}
\cdot
\bigl(\mathbf h_R^{(i-1)}\bigr)^{\mathbf b_L^{(i-1)}}
\cdot u^{c_R}.
\end{aligned}
\]
$\mathcal P$ sends $L_i,R_i$ to $\mathcal V$.

\item
$\mathcal V$ samples $x_i\xleftarrow{\$}\mathbb F_p^*$
and sends $x_i$ to $\mathcal P$.

\item
$\mathcal P$ and $\mathcal V$ both compute
\[
P^{(i)}
=
L_i^{x_i^2}\cdot P^{(i-1)}\cdot R_i^{x_i^{-2}},
\]
and
\[
\begin{aligned}
\mathbf g^{(i)}
&=
\bigl(\mathbf g_L^{(i-1)}\bigr)^{x_i^{-1}}
\circ
\bigl(\mathbf g_R^{(i-1)}\bigr)^{x_i},
\\
\mathbf h^{(i)}
&=
\bigl(\mathbf h_L^{(i-1)}\bigr)^{x_i}
\circ
\bigl(\mathbf h_R^{(i-1)}\bigr)^{x_i^{-1}}.
\end{aligned}
\]

\item
$\mathcal P$ computes
\[
\begin{aligned}
\mathbf a^{(i)}
&=
x_i\mathbf a_L^{(i-1)}
+x_i^{-1}\mathbf a_R^{(i-1)},
\\
\mathbf b^{(i)}
&=
x_i^{-1}\mathbf b_L^{(i-1)}
+x_i\mathbf b_R^{(i-1)}.
\end{aligned}
\]
\end{enumerate}

After $k$ rounds, $\mathcal P$ sends
$a^{(k)},b^{(k)}\in\mathbb F_p$ to $\mathcal V$,
identifying length-one vectors with their entries.

\medskip
\textbf{Verification.}
$\mathcal V$ accepts if and only if
\[
P^{(k)}
=
\bigl(g^{(k)}\bigr)^{a^{(k)}}
\cdot
\bigl(h^{(k)}\bigr)^{b^{(k)}}
\cdot
u^{a^{(k)}b^{(k)}}.
\]

\end{minipage}}
\caption{Bulletproofs' Inner Product Argument (BP-IPA).}
\label{fig:bp-ipa}
\end{figure}

\subsection{Warm-Up: The Bulletproofs Instantiation}
\label{sec:bp-warmup}

Before the instantiation of $\Pi_{\mathsf{IPA}}$, we analyse a simpler one, in which the
inner-product argument is the transparent argument of Bulletproofs
(\cref{fig:bp-ipa}) over uniformly sampled generators.
We present it as a warm-up for two reasons.
The construction is a little simpler: the commitment key has no structure, no
pairing is involved, and a single commitment absorbs both vectors and their
inner product.
The analysis is also not specific to one scheme.
It uses only the recursive halving and the fact that each round's messages are
fixed functions of a fresh challenge, so it applies to Bulletproofs-like
arguments in general.
We stress that our construction for TRACE uses the other instantiation, with a
Pedersen commitment under a structured key and the argument of Daza, R\`afols
and Zacharakis (\cref{fig:pv-ipa}), and \cref{sec:ipa-nonamort-proof} transfers
the analysis to it.

\paragraph{Prover experiment for the inner-product argument.}
Fix $N=2^r\ge2$ and instantiate the generic group model above with the
$2N+1$ generators $\mathbf g,\mathbf h,u$, sampled independently of the
witness and of the challenge oracle.
We write $r$ for the number of rounds, denoted $k$ in \cref{fig:bp-ipa}.
Throughout this subsection $n$ denotes the number of identities.
A generic prover receives $(\mathbf a,\mathbf b)$ in the clear together with
the handles of the identity, of $\mathbf g,\mathbf h,u$, and of
$P=\mathbf g^{\mathbf a}\mathbf h^{\mathbf b}
u^{\langle\mathbf a,\mathbf b\rangle}$.
These are its only initial group handles.
It then runs the protocol of \cref{fig:bp-ipa} against the honest verifier.
We write $Q_{\mathrm{total}}$ for the number of oracle queries it makes, so
that every group multiplication and division is counted, including
preprocessing, and a scalar exponentiation is charged for the group
operations that implement it.
Challenges are uniform in $\mathbb F_p^\times$ and fresh, either because the
verifier samples them or because they are derived by full-transcript
Fiat--Shamir in the random oracle model.
We write $\mathsf{Acc}$ for the event that the verifier accepts.

\begin{theorem}[Group-operation lower bound for BP-IPA]
\label{thm:ipa-lower-bound}
Let $N=2^r\ge2$ and let $N$ and $\log p$ be polynomially bounded in
$\lambda$.
Suppose the joint distribution of $(\mathbf a,\mathbf b)\in\mathbb F_p^{2N}$
has HILL min-entropy at least $K$, and assume
$K,\log p=\omega(\log\lambda)$.
Then for every PPT generic prover $\mathcal A$ in the experiment above there
is a constant $c_{\mathcal A}>0$ such that
\[
    \Pr\!\left[
      \mathsf{Acc}\ \wedge\
      Q_{\mathrm{total}}<c_{\mathcal A}
          \left(N+\frac{K+(r-1)\log p}{\log\lambda}\right)
    \right]
    \le\operatorname{negl}(\lambda).
\]
\end{theorem}

\begin{proof}
Fix $\lambda$ and a PPT generic prover $\mathcal A$.
All logarithms are binary.
We separate three sources of work.
The first messages encode the witness and therefore incur a cost determined
by $K$.
Each later message pair depends on a fresh challenge, which incurs an
additional cost determined by $\log p$.
Finally, the last pair must involve many of the original generators, which
gives the $N$ term.
We first explain why an arbitrary accepting prover must satisfy the
coefficient constraints used in these arguments.

\emph{Formal computation and fresh challenges.}
Associate the $2N+1$ initial generators with the standard basis of
$\mathbb F_p^{2N+1}$, and associate $P$ with
$(\mathbf a,\mathbf b,\langle\mathbf a,\mathbf b\rangle)$.
Multiplication and division add and subtract formal vectors.
The symbolic group oracle assigns the same handle to equal vectors and
uniformly distinct random handles to distinct vectors.
As in the proof of \cref{thm:msm-entropy-lower-bound}, the real and symbolic
experiments differ only if distinct formal vectors evaluate to the same
group element, which for each fixed nonzero difference has probability at
most $1/(p-1)$ over the independent random generators.
A union bound over the polynomially many vectors considered by the prover
and the verifier bounds the difference by $\operatorname{poly}(\lambda)/(p-1)$,
which is negligible, so it suffices to work symbolically.

For Fiat--Shamir, let $\tau_{i-1}$ contain all earlier messages and
challenges and set
\[
    x_i=H(\mathsf{ipa},\mathsf{stmt},i,\tau_{i-1},L_i,R_i)
          \in\mathbb F_p^\times,
    \qquad
    \mathsf{stmt}=(N,\mathbf g,\mathbf h,u,P),
\]
using an unambiguous encoding.
We append the verifier's evaluation of any unqueried challenges to the
prover's oracle-query log, before the verifier's group computations.
Except with negligible probability, the first queries for the prefixes of
the submitted transcript occur in protocol order: a later prefix queried
first would have to guess the answer to an earlier fresh query, and a
relevant query cannot contain a group handle first obtained later except by
guessing a fresh encoding.
A union bound over the polynomially many queries and handles bounds both
events by $\operatorname{poly}(\lambda)/(p-1)$.

Consequently we can divide an execution into the interval before the first
relevant challenge query and the intervals between successive relevant
challenge queries.
The interactive case has this ordering automatically.
In either case, conditional on the history immediately before a fresh
challenge is sampled, that challenge is uniform.
Whenever we bound an exceptional event at a fixed prefix below, a union
bound over all polynomially many prefixes also covers the prover's adaptive
choice of which transcript to submit.

\emph{Acceptance forces the prescribed message coefficients.}
Project $\mathbb F_p^{2N+1}$ onto its $\mathbf g,\mathbf h$ coordinates,
discarding the $u$ coordinate.
Let $S_i$ be the span of the projected folded generators after round $i$.
The folded generators are formally linearly independent, and
\[
    S_r\subseteq S_{r-1}\subseteq\cdots\subseteq S_0=\mathbb F_p^{2N}.
\]
Write $p_i,\ell_i,\rho_i$ for the projections of $P^{(i)},L_i,R_i$, so that
\[
    p_i=p_{i-1}+x_i^2\ell_i+x_i^{-2}\rho_i .
\]
The final verification equation implies $p_r\in S_r$.

Suppose $p_i\notin S_i$ at some earlier round.
Since acceptance requires $p_r\in S_r\subseteq S_i$, there is a later round
$j$ at which the commitment's image in the quotient by $S_i$ first becomes
zero.
At that round
\[
    \overline p_{j-1}+x_j^2\overline\ell_j+x_j^{-2}\overline\rho_j=0,
    \qquad \overline p_{j-1}\ne0 ,
\]
and all three coefficients are fixed before $x_j$ is sampled.
Multiplying by $x_j^2$ gives a nonzero vector polynomial of degree at most
four, so this event has probability at most $4/(p-1)$.
A union bound over prefixes and their ancestors gives
\begin{equation}
    \Pr[\mathsf{Acc}\ \wedge\ \exists i:\ p_i\notin S_i]
    \le\operatorname{negl}(\lambda).
    \label{eq:ipa-lb-membership}
\end{equation}

We now proceed inductively.
Suppose the previous messages have their prescribed
$\mathbf g,\mathbf h$ coefficients, so that the current commitment has
witness coefficients $(a_L,a_R,b_L,b_R)$ in the current folded basis.
For fixed proposed messages, membership
\[
    p_{i-1}+x^2\ell_i+x^{-2}\rho_i\in S_i(x)
\]
is a system of polynomial equations of degree at most six after clearing
denominators.
If these equations hold identically, the components outside $S_{i-1}$ first
force $\ell_i,\rho_i\in S_{i-1}$.
Inside that space the generator-folding constraints give
\begin{align*}
    &x^6\ell_{g,L}+x^4(a_L-\ell_{g,R})
        +x^2(\rho_{g,L}-a_R)-\rho_{g,R}=0,\\
    &-x^6\ell_{h,R}+x^4(\ell_{h,L}-b_R)
        +x^2(b_L-\rho_{h,R})+\rho_{h,L}=0 .
\end{align*}
They are identities precisely when, in the current basis,
\begin{equation}
    \ell_i=(0,a_L,b_R,0),
    \qquad
    \rho_i=(a_R,0,0,b_L).
    \label{eq:ipa-lb-prescribed}
\end{equation}
Otherwise at least one nonzero polynomial of degree at most six vanishes at
the fresh challenge, which happens with probability at most $6/(p-1)$.
Considering the first round at which \eqref{eq:ipa-lb-prescribed} fails and
using \eqref{eq:ipa-lb-membership} shows that every accepting transcript has
the prescribed coefficients in every round, except with negligible
probability.
No restriction on the messages' $u$ coefficients is needed for this
conclusion.

We also exclude vanishing folded witnesses.
HILL indistinguishability, applied to the efficient zero test, gives
\[
    \Pr[(\mathbf a,\mathbf b)=0]\le2^{-K}+\operatorname{negl}(\lambda)
    =\operatorname{negl}(\lambda).
\]
If a fixed nonzero vector is folded as $xz_L+x^{-1}z_R$, obtaining zero
requires the nonzero coordinate polynomial $x^2z_L+z_R$ to vanish, which
happens for at most two nonzero values of $x$.
The reversed folding rule has the same property.
A union bound therefore excludes a nonzero vector becoming zero along the
folding process.
We use this both for the current witness and, later, for its restrictions to
either of the two supports of the final folded generators; each such
restriction itself follows a sequence of the same folding operations.

\emph{The first messages incur the entropy cost.}
Starting from $M$ handles, a computation using at most $s$ group operations
and selecting two output handles has at most
\begin{equation}
    \Gamma_s(M)=(M+s)^2\bigl(2(M+s)^2\bigr)^s
    \label{eq:ipa-lb-traces}
\end{equation}
possible descriptions: each instruction selects two available handles and a
multiplication or a division, and the factor $(M+s)^2$ accounts for the
outputs.
Shorter computations are padded with operations on the identity.
This is the counting bound of \cref{thm:msm-entropy-lower-bound} with the
output widened to a pair.

Initially $M_0=2N+3$.
For a fixed syntactic trace over these input records, the two outputs have
expressions
\[
    L_1=P^\alpha\mathbf g^{d_g}\mathbf h^{d_h}u^{d_u},
    \qquad
    R_1=P^\beta\mathbf g^{e_g}\mathbf h^{e_h}u^{e_u},
\]
whose exponents depend only on the trace.
Set
\[
    v=(a_L,a_R,b_L,b_R),\qquad
    J_Lv=(0,a_L,b_R,0),\qquad
    J_Rv=(a_R,0,0,b_L).
\]
The prescribed first messages satisfy
\begin{equation}
    (J_L-\alpha I)v=d,
    \qquad (J_R-\beta I)v=e,
    \qquad d=(d_g,d_h),\ e=(e_g,e_h).
    \label{eq:ipa-lb-decode}
\end{equation}
These equations determine at most one witness pair.
If $\alpha\ne0$ then $J_L-\alpha I$ is invertible, because $J_L^2=0$ makes
$J_L$ nilpotent; the same applies to $J_R-\beta I$ if $\beta\ne0$.
If both are zero, adding the equations determines $v$, since $J_L+J_R$
permutes the four blocks and is therefore invertible.

Thus the set of witnesses for which some trace of length at most $s$
produces the prescribed first messages has size at most $\Gamma_s(M_0)$, and it
is independent of the oracle randomness.
On a distribution of min-entropy at least $K$ its probability is at most
$\Gamma_s(M_0)2^{-K}$.
To transfer this to HILL entropy, an efficient distinguisher, given a
candidate witness, simulates the symbolic group and the random oracle, runs
$\mathcal A$, and checks whether its first messages have the prescribed
coefficients and were obtained within the first $s$ group operations.
This tests the actual execution and does not enumerate short traces.
Writing $q_0$ for the number of group operations before the first relevant
challenge, acceptance, the coefficient condition and the ordering of
relevant queries give
\begin{equation}
    \Pr[\mathsf{Acc}\ \wedge\ q_0\le s]
    \le \Gamma_s(M_0)2^{-K}+\operatorname{negl}(\lambda).
    \label{eq:ipa-lb-first}
\end{equation}

\emph{Each subsequent pair incurs a fresh-challenge cost.}
Fix the history immediately before challenge $x_i$ is sampled, where $i<r$,
and suppose the current witness is nonzero.
Substituting the witness and generator folds into the prescribed next
messages gives a pair of formal coefficient vectors
\begin{equation}
    F_i(x)=A_ix^2+B_i+C_ix^{-2},
    \label{eq:ipa-lb-laurent}
\end{equation}
with $A_i,B_i,C_i$ fixed by this history.
This function is nonconstant.
To see it explicitly, split the current vectors into four blocks and put
\[
    D_{jk}=\operatorname{coeff}_{g,h}
           \bigl(\mathbf g_k^{\mathbf a_j}\mathbf h_j^{\mathbf b_k}\bigr),
\]
where the bases and witnesses are the current ones, before folding.
Expanding the next messages yields
\[
    A_i=(D_{14},D_{23}),\quad
    B_i=(D_{12}+D_{34},D_{21}+D_{43}),\quad
    C_i=(D_{32},D_{41}).
\]
Because the current generators are formally independent, $A_i=C_i=0$ would
force all four blocks of both witness vectors to be zero.

For each fixed target $z$ the equation $F_i(x)=z$ implies
\[
    A_ix^4+(B_i-z)x^2+C_i=0 ,
\]
in which some coordinate is a nonzero polynomial of degree at most four, so
at most four values of $x\in\mathbb F_p^\times$ map to $z$.
Hence, for every fixed prechallenge history,
\[
    \max_z\Pr[F_i(x_i)=z\mid\text{history}]\le\frac4{p-1},
\]
that is, the target pair has min-entropy at least $\log(p-1)-2$ before the
fresh challenge is sampled.
If the history contains $M$ handles, at most $\Gamma_s(M)$ fixed output pairs are
obtainable by traces of $s$ additional group operations, so
\begin{equation}
    \Pr[\text{prescribed next pair obtainable within $s$ operations}
        \mid\text{history}]
    \le\frac{4\Gamma_s(M)}{p-1}.
    \label{eq:ipa-lb-transition}
\end{equation}
Arbitrary scalar computation and additional oracle queries do not increase
the number of possible group-operation traces.
Moreover \eqref{eq:ipa-lb-transition} holds before conditioning on whether
this prefix is eventually selected, so a union bound over all fresh queries
accounts for grinding and adaptive transcript selection.

\emph{The final messages must involve many original generators.}
Let $S,T\subseteq[N]$ be the supports of the two $\mathbf g$-generators just
before the last round.
They partition $[N]$, each has size $N/2$, and every coefficient on either
support is nonzero.
Write $G_S,G_T$ for these two formal vectors and
$\widehat a_L,\widehat a_R$ for the corresponding folded witness entries, so
that $\widehat a_L$ is obtained by folding the restriction $\mathbf a_S$ and
$\widehat a_R$ by folding $\mathbf a_T$.
By the nonvanishing argument above, each is nonzero whenever its original
restriction is nonzero, except with negligible probability.
Assume first that $\mathbf a\ne0$; if instead $\mathbf a=0$ then
$\mathbf b\ne0$ and the same argument uses the $\mathbf h$ coordinates.
The prescribed $\mathbf g$ coefficients of the last messages are
\begin{equation}
    \ell_r^{g}=\widehat a_LG_T,
    \qquad
    \rho_r^{g}=\widehat a_RG_S .
    \label{eq:ipa-lb-final}
\end{equation}

Fix syntactic expressions for the actual computation over its initial input
records, and let $I\subseteq[N]$ contain the indices of the original $g_j$
records used as operands, together with those selected directly as either of
the last two messages.
Since $P$ has $\mathbf g$ coefficient vector $\mathbf a$, and the other
initial handles have singleton or zero $\mathbf g$ support, there are
scalars $\gamma,\delta$ and vectors $d,e$ with
\begin{equation}
    \ell_r^{g}=\gamma\mathbf a+d,
    \qquad
    \rho_r^{g}=\delta\mathbf a+e,
    \qquad
    \operatorname{supp}(d),\operatorname{supp}(e)\subseteq I .
    \label{eq:ipa-lb-support}
\end{equation}
We show $|I|\ge N/2$ in two cases.

\emph{Case 1: exactly one of $\mathbf a_S,\mathbf a_T$ is nonzero.}
Suppose $\mathbf a_S\ne0$ and $\mathbf a_T=0$.
Then $\widehat a_L\ne0$, so $\ell_r^g$ is nonzero at every coordinate of
$T$, while $\gamma\mathbf a$ is zero throughout $T$.
Equation~\eqref{eq:ipa-lb-support} therefore forces $T\subseteq I$.
The other case is symmetric, using $\rho_r^g$.

\emph{Case 2: both $\mathbf a_S$ and $\mathbf a_T$ are nonzero.}
Then both folded scalars are nonzero.
Suppose there were $j\in S\setminus I$ and $k\in T\setminus I$.
At coordinate $j$, equations~\eqref{eq:ipa-lb-final}
and~\eqref{eq:ipa-lb-support} give $0=\gamma a_j$ and $0\ne\delta a_j$,
hence $\gamma=0$.
At coordinate $k$ the same equations give $0\ne\gamma a_k$, a contradiction.
Thus $S\subseteq I$ or $T\subseteq I$, again giving $|I|\ge N/2$.

Each binary group operation uses at most two original generator records as
operands, and at most two further records are selected directly as the last
messages, so
\begin{equation}
    \frac N2\le|I|\le2Q_{\mathrm{total}}+2,
    \qquad
    Q_{\mathrm{total}}\ge\frac N4-1,
    \label{eq:ipa-lb-linear}
\end{equation}
on every accepting execution outside the negligible exceptional events.
This argument includes all preprocessing: a generator used to form a cached
handle has already entered the counted computation.

\emph{Combining the bounds.}
Put $\ell=\log\lambda$ and $m=\log p$.
There is a constant $C>0$, depending on the polynomial bounds for
$\mathcal A$ and the parameters, with $\log \Gamma_s(M)\le C(s+1)\ell$ for all
handle counts and cutoffs used below.
Set
\[
    t_0=\left\lfloor\frac K{4C\ell}\right\rfloor,
    \qquad
    t_1=\left\lfloor\frac m{4C\ell}\right\rfloor .
\]
Since $K,m=\omega(\ell)$, \eqref{eq:ipa-lb-first} bounds the probability of
acceptance with $q_0\le t_0$ by $2^{-K/2}+\operatorname{negl}(\lambda)$ for
all sufficiently large $\lambda$.
Likewise \eqref{eq:ipa-lb-transition}, followed by a union bound over
polynomially many queries, bounds the probability of an accepting transcript
with any transition costing at most $t_1$ operations by
$\operatorname{poly}(\lambda)2^{-m/2}+\operatorname{negl}(\lambda)$.
Both quantities are negligible.

Let $q_i$ count prover group operations between the relevant challenge
queries $x_i$ and $x_{i+1}$, for $1\le i<r$.
If a challenge is first evaluated by the verifier after the proof is output,
no further prover operations occur in that interval.
The intervals counted by $q_0,q_1,\ldots,q_{r-1}$ are disjoint, so except
with negligible probability acceptance implies
\[
    Q_{\mathrm{total}}
    \ge q_0+\sum_{i=1}^{r-1}q_i
    >t_0+(r-1)t_1
    =\Omega\!\left(\frac{K+(r-1)\log p}{\log\lambda}\right).
\]
There is no further interval to charge after the last challenge, because the
last prover message consists only of scalars.

The support argument may count operations already included in these
intervals, so we combine it by taking a maximum:
\[
    Q_{\mathrm{total}}
    \ge\max\left\{\frac N4-1,\ t_0+(r-1)t_1\right\}
    =\Omega\!\left(N+\frac{K+(r-1)\log p}{\log\lambda}\right),
\]
where the last step uses $K/\log\lambda\to\infty$ to absorb the additive
constant and the fact that the maximum of two nonnegative quantities is at
least half their sum.
All exceptional events above have negligible total probability, which proves
the claimed statement.
\end{proof}

\paragraph{The warm-up construction.}
Write $\Pi^{\mathsf{BP}}_{\mathsf{IPA}}$ for the argument with registered prover
obtained from \cref{app:construction} by using the argument of
\cref{fig:bp-ipa} over uniformly sampled generators
$\mathbf g,\mathbf h\in\mathbb G^N$ and $u\in\mathbb G$.
$\setup$ samples these generators, and we write $H=\prod_{i\le N}\mathbf h_i$.
$\KG(\pcrs)$ outputs a signature key pair with $\pk=\mathsf{svk}$.
Messages are traces $E\in\mathbb F_p^N$ with $\Com(\pcrs,E)=\mathbf g^{E}$.
The prover sets $v=\langle E,\mathbf 1\rangle$, so that the statement
$P=C_E\cdot H\cdot u^{v}
=\mathbf g^{E}\mathbf h^{\mathbf 1}u^{\langle E,\mathbf 1\rangle}$ of
$\rel_{\mathrm{BP\text{-}IPA}}$ is computable from public data.
It derives $\theta=\ro(\mathsf{claim},C_E,\pk,v,\nonce)$ and runs the prover of
\cref{fig:bp-ipa} on $(E,\mathbf 1)$ with the Fiat--Shamir oracle $\ro_\theta$.
The signature of \cref{app:construction} plays no role in non-amortizability,
since it only adds work, and we omit it.

\begin{proposition}[Warm-up: non-amortizability of $\Pi^{\mathsf{BP}}_{\mathsf{IPA}}$]
\label{prop:bp-nonamort}
Let $N=2^r$ and $\log p$ be polynomially bounded in $\lambda$ with
$\log p=\omega(\log\lambda)$, and let $\Delta$ be as in
\cref{thm:ipa-nonamort}.
In the generic group model with a random oracle,
$\Pi^{\mathsf{BP}}_{\mathsf{IPA}}$ satisfies non-amortizability with overhead
$\Delta$ over every distribution $\xwDist$.
\end{proposition}

\begin{proof}
Fix $n=\operatorname{poly}(\lambda)$ and an adversary $\adv$ with
$\mathsf{Time}(\adv)/n=o(\Delta)$.
Let $Q$ be the number of group oracle queries it makes and $q_H$ the number
of random oracle queries, so $Q,q_H\le\mathsf{Time}(\adv)$ and both are
polynomially bounded.
Write $M_0$ for the number of initial handles and $M=M_0+Q$ for the number
of handles ever issued, so $\log M,\log q_H=O(\log\lambda)$.
Set $m=n(r-1)$.

\emph{Symbolic execution and independent challenge chains.}
All $n$ proofs are produced against one generic group, so the symbolic
simulation and the collision bound of \cref{thm:ipa-lower-bound} apply
verbatim to the whole execution.
The public keys $\pk_1,\ldots,\pk_n$ are distinct except with negligible
probability, so the seeds $\theta_j$ are distinct, and each of the $nr$
challenges $x_i^{(j)}$ is a random oracle answer on an input of its own.
Since $\adv$ chooses the nonce, and with it the seeds, it may query the
oracle on many candidate nonces and submit proofs for the one it prefers, so
the submitted challenges are selected answers rather than independent uniform
values.
The counting step below accounts for this selection by a union bound over
which oracle answers play the role of the challenges.
Each witness $(E,\mathbf 1)$ is nonzero, since its second component is, and the
nonvanishing argument of
\cref{thm:ipa-lower-bound} applies to each of the $n$ folding processes.
No entropy assumption on $\xwDist$ is needed: the adversary is given the
secret keys, so all unpredictability that the argument exploits comes from
the random oracle.

\emph{Prescribed coefficients.}
Applying the step \emph{Acceptance forces the prescribed message
coefficients} of \cref{thm:ipa-lower-bound} to each of the $n$ transcripts
separately, and taking a union bound over $j$, except with negligible
probability every accepting transcript satisfies, for all $j\in[n]$ and all
$i\in[r-1]$,
\begin{equation}
    \bigl(\ell^{(j)}_{i+1},\rho^{(j)}_{i+1}\bigr)
    =F_{j,i}\bigl(x^{(j)}_i\bigr),
    \label{eq:na-transition}
\end{equation}
where $F_{j,i}$ is the nonconstant Laurent polynomial
\eqref{eq:ipa-lb-laurent} determined by $(E,\mathbf 1)$, the commitment
key and $x^{(j)}_1,\ldots,x^{(j)}_{i-1}$.

\emph{Global counting.}
The interval decomposition used in \cref{thm:ipa-lower-bound} does not carry
over, because $\adv$ may interleave the $n$ executions arbitrarily and an
interval opened by one challenge may contain work charged to another.
We therefore count the whole execution at once.
A \emph{description} consists of the indices of the two input handles and
the sign of each of the $Q$ group operations, together with the indices of
the two handles carrying the messages
$\bigl(L^{(j)}_{i+1},R^{(j)}_{i+1}\bigr)$ for each of the $m$ pairs.
The number of descriptions is at most
\[
    \bigl(2M^2\bigr)^{Q}M^{2m}.
\]
Fix a description.
Since the initial handles have formal vectors determined by the commitment
key, by $E$ and by the secret keys, all of which $\adv$ receives, the
description determines the formal vector of every handle it names.
Reconstruct the challenges as follows: for each $j$, and for
$i=1,\ldots,r-1$ in order, the polynomial $F_{j,i}$ is determined by the
already reconstructed $x^{(j)}_1,\ldots,x^{(j)}_{i-1}$, and
\eqref{eq:na-transition} clears to a nonzero polynomial equation of degree
at most four in $x^{(j)}_i$, which has at most four solutions in
$\mathbb F_p^\times$.
Hence at most $4^m$ challenge tuples are consistent with a single
description.
Taking a union bound over which $m$ of the $q_H$ random oracle answers play
the role of the challenges, and using that those answers are independent and
uniform,
\[
    \Pr[\textstyle\bigwedge_j\vfy=1\ \wedge\ Q\le q]
    \le\frac{q_H^m\bigl(2M^2\bigr)^{q}M^{2m}4^m}{(p-1)^m}
      +\operatorname{negl}(\lambda).
\]
Taking logarithms, there is a constant $C'>0$ with
\[
    \log\Pr\le C'(q+m)\log\lambda-m\log(p-1)
              +\operatorname{negl}(\lambda).
\]
Since $\log p=\omega(\log\lambda)$, choosing
$q=\bigl\lfloor m\log p/(2C'\log\lambda)\bigr\rfloor-m$ makes the right-hand
side at most $-m\log p/2$ for all sufficiently large $\lambda$, and
$m\log p\ge\log p=\omega(\log\lambda)$, so this probability is negligible.
Therefore, except with negligible probability, an execution in which all $n$
proofs verify has
\[
    Q=\Omega\!\left(\frac{m\log p}{\log\lambda}\right)
     =\Omega\!\left(\frac{n(r-1)\log p}{\log\lambda}\right)
     =\Omega(n\Delta).
\]

\emph{Conclusion.}
Every group oracle query costs one step, so
$\mathsf{Time}(\adv)\ge Q$.
The previous paragraph gives a constant $c>0$ such that an execution in
which all $n$ proofs verify has $Q\ge cn\Delta$, except with negligible
probability, whereas $\mathsf{Time}(\adv)/n=o(\Delta)$ gives
$\mathsf{Time}(\adv)<cn\Delta$ for all sufficiently large $\lambda$.
The two are inconsistent, so
$\Pr[\bigwedge_j\vfy^\ro(\vcrs,\pk_j,\Com(\pcrs,\vec a),\nonce,\pi_j)=1]
\le\operatorname{negl}(\lambda)$.
\end{proof}

\paragraph{What the bound does and does not say.}
Using $\log p=\Omega(\lambda)$ from the generic group model above,
\cref{prop:bp-nonamort}, and with it \cref{thm:ipa-nonamort}, gives, for $N\ge4$, a marginal cost per identity of
\[
    \Delta=\Omega\!\left(\frac{\log N\cdot\lambda}{\log\lambda}\right),
\]
which grows with the trace length and with the security parameter and is
therefore a genuine per-identity charge.
It is nevertheless far below the $\Theta(N\log p)$ group operations the
honest prover performs, and the gap is not an artifact of the proof.
Neither of the two other terms in \cref{thm:ipa-lower-bound} survives the
experiment.
The $N$ term of \eqref{eq:ipa-lb-linear} charges the prover for touching
half of the generator handles, and once touched those handles remain
available, so the whole charge is paid by the first proof.
The $K$ term charges the first message pair for encoding the witness, and
no key enters that pair:
$L_1=\mathbf g_R^{E_L}\mathbf h_L^{\mathbf 1}u^{\langle E_L,\mathbf 1\rangle}$
depends only on the trace and the commitment key, so all $n$ identities share
it.
\Cref{app:amortization} shows how much of the later rounds can be shared as
well, for the argument of \cref{fig:pv-ipa}.
We therefore do not claim that \cref{app:cost}'s $\Theta(\ell)$ per identity
is necessary; that is a statement about the honest prover.
Closing the gap would require a lower bound for multi-scalar
multiplications with structured exponents and polynomial preprocessing,
which neither the counting argument nor the support argument can provide.

\subsection{A Publicly Verifiable Inner-Product Argument}
\label{sec:pv-ipa}

\Cref{fig:pv-ipa} is the argument $\Pi_{\mathsf{IPA}}$ uses in place of the
transparent recursion of \cref{fig:bp-ipa}, and it is what keeps the on-chain
verifier logarithmic.
Its relation is $\rel_{\mathsf{IP}}[\mathsf{ck}_{\mathsf P}]$ of
\cref{app:ingredients}.
Unlike $\rel_{\mathrm{BP\text{-}IPA}}$, the two vectors are committed under
separate structured keys, and the inner product $z$ is a public field element
instead of being absorbed into a $u^{\langle\mathbf a,\mathbf b\rangle}$ factor.
In both instantiations the key enters only through the Fiat--Shamir label.
\Cref{thm:ipa-nonamort} is proved for this argument in
\cref{sec:ipa-nonamort-proof}, starting from the warm-up of
\cref{sec:bp-warmup}.

\begin{figure}[tp]
\centering
\fbox{\begin{minipage}{\dimexpr\linewidth-2\fboxsep-2\fboxrule\relax}
\scriptsize

\textbf{Setting.}
Let $e:\mathbb G_1\times\mathbb G_2\to\mathbb G_T$ be an asymmetric bilinear
map over groups of prime order $p$ with fixed generators, and write $[a]_t$
for the encoding of $a\in\mathbb F_p$ in $\mathbb G_t$, so that
$e([a]_1,[b]_2)=[ab]_T$.
For $\mathbf a\in\mathbb F_p^{N}$ and a key $[\mathbf r]_1\in\mathbb G_1^{N}$,
write $[\mathbf a^{\top}\mathbf r]_1=\sum_i a_i[r_i]_1$ for the vector
Pedersen commitment of $\mathbf a$.
Let $N=2^{\nu}$, and as in \cref{fig:bp-ipa} let $\mathbf v_L,\mathbf v_R$ be
the halves of an even-length vector $\mathbf v$.

\medskip
\textbf{Structured commitment keys.}
The key is drawn with multiplicative structure, either
$\mathcal{PW}_n:\mathbf r=(1,x,x^2,\ldots,x^{N-1})$, a powers-of-tau string,
or
$\mathcal{ML}_{2^\nu}:\mathbf r=(1,x_1,x_2,x_1x_2,\ldots,x_1\cdots x_\nu)$,
the multilinear monomials; the first is the case $x_i=x^{2^{i-1}}$.
The prover holds the full key $[\mathbf r]_1$, and the verifier holds only its
compressed description $[\mathbf x]_2=([x_1]_2,\ldots,[x_\nu]_2)$, of size
$\log N$, together with the current generator $[r]_1$.
Write
\[
\begin{gathered}
    \bigl(\mathsf{pp},[r]_1,[\mathbf r]_1,[\mathbf x]_2\bigr)
    \in\mathcal L_{\mathsf{Com}}
    \iff
    [r_1]_1=[r]_1
    \ \wedge \ 
    [r_{2^{i-1}+j}]_1=x_i[r_j]_1
    \quad
    \forall i\in[\nu],\ j\in\bigl[2^{i-1}\bigr].
\end{gathered}
\]

\medskip
\textbf{Relation.}
\[
\begin{gathered}
\bigl(\mathsf{pp},[r]_1,[s]_1,[\mathbf x]_2,[\mathbf y]_2,
       [\alpha]_1,[\beta]_1,z\bigr)\in\mathcal L_{\mathrm{IP}}
\iff
\exists\,[\mathbf r]_1,[\mathbf s]_1\in\mathbb G_1^{N},\
\mathbf a,\mathbf b\in\mathbb F_p^{N}:\\[2pt]
\bigl(\mathsf{pp},[r]_1,[\mathbf r]_1,[\mathbf x]_2\bigr)
  \in\mathcal L_{\mathsf{Com}}
\ \wedge\
\bigl(\mathsf{pp},[s]_1,[\mathbf s]_1,[\mathbf y]_2\bigr)
  \in\mathcal L_{\mathsf{Com}}
\ \wedge\\[2pt]
[\alpha]_1=[\mathbf a^{\top}\mathbf r]_1
\ \wedge\
[\beta]_1=[\mathbf b^{\top}\mathbf s]_1
\ \wedge\
\mathbf a^{\top}\mathbf b=z .
\end{gathered}
\]

\medskip
\textbf{Reduction.}
One round reduces a statement of size $N=2^\nu$ to one of size $N/2$.
Both parties hold
$(\mathsf{pp},[r]_1,[s]_1,[\mathbf x]_2,[\mathbf y]_2,[\alpha]_1,[\beta]_1,z)$;
the prover also holds $[\mathbf r]_1,[\mathbf s]_1,\mathbf a,\mathbf b$.
\begin{enumerate}[
    label=\arabic*.,
    leftmargin=*,
    itemsep=0.4em,
    topsep=0.3em
]
\item
$\mathcal P$ sends the cross terms
\[
\begin{gathered}
[\alpha_{-1}]_1=[\mathbf a_L^{\top}\mathbf r_R]_1,
\qquad
[\alpha_{1}]_1=[\mathbf a_R^{\top}\mathbf r_L]_1,
\qquad
z_{-1}=\mathbf a_R^{\top}\mathbf b_L,
\\
[\beta_{-1}]_1=[\mathbf b_L^{\top}\mathbf s_R]_1,
\qquad
[\beta_{1}]_1=[\mathbf b_R^{\top}\mathbf s_L]_1,
\qquad
z_{1}=\mathbf a_L^{\top}\mathbf b_R .
\end{gathered}
\]

\item
$\mathcal V$ samples $c\sample\mathbb F_p^{\times}$ and sends it to
$\mathcal P$.

\item
$\mathcal P$ folds the witness and both keys,
\[
\begin{gathered}
\mathbf a'=c\,\mathbf a_L+c^{2}\mathbf a_R,
\qquad
\mathbf b'=c^{-1}\mathbf b_L+c^{-2}\mathbf b_R,
\\
[\mathbf r']_1=c^{-1}[\mathbf r_L]_1+c^{-2}[\mathbf r_R]_1,
\qquad
[\mathbf s']_1=c\,[\mathbf s_L]_1+c^{2}[\mathbf s_R]_1,
\end{gathered}
\]
and sends $[r']_1$ and $[s']_1$, the first components of $[\mathbf r']_1$ and
$[\mathbf s']_1$, which are the generators of the reduced keys.

\item
$\mathcal V$ checks the two pairing equations
\[
    e\bigl([r']_1-c^{-1}[r]_1,[1]_2\bigr)
    =e\bigl(c^{-2}[r]_1,[x_\nu]_2\bigr),
    \qquad
    e\bigl([s']_1-c\,[s]_1,[1]_2\bigr)
    =e\bigl(c^{2}[s]_1,[y_\nu]_2\bigr),
\]
and aborts if either fails.
This replaces recomputing the folded key.

\item
Both compute
\[
\begin{gathered}
[\mathbf x']_2=([x_i]_2)_{i<\nu},
\qquad
[\mathbf y']_2=([y_i]_2)_{i<\nu},
\\
[\alpha']_1=c^{-1}[\alpha_{-1}]_1+[\alpha]_1+c\,[\alpha_{1}]_1,
\qquad
[\beta']_1=c\,[\beta_{-1}]_1+[\beta]_1+c^{-1}[\beta_{1}]_1,
\\
z'=c\,z_{-1}+z+c^{-1}z_{1},
\end{gathered}
\]
and recurse on
$(\mathsf{pp},[r']_1,[s']_1,[\mathbf x']_2,[\mathbf y']_2,
  [\alpha']_1,[\beta']_1,z')$.
\end{enumerate}

\medskip
\textbf{Verification.}
After $\nu$ rounds both vectors have length one.
$\mathcal P$ sends $a,b\in\mathbb F_p$ and $\mathcal V$ accepts if and only if
\[
    [\alpha]_1=a[r]_1,
    \qquad
    [\beta]_1=b[s]_1,
    \qquad
    ab=z .
\]

\end{minipage}}
\caption{The publicly verifiable inner-product argument of Daza, R\`afols and
Zacharakis~\cite{drz}, in interactive form; the challenges are derived by
Fiat--Shamir in $\prv$ (\cref{app:construction}).
The verifier never reads the full commitment key.
It holds the $\log N$ elements $[\mathbf x]_2,[\mathbf y]_2$ and runs in logarithmic time.}
\label{fig:pv-ipa}
\end{figure}

\subsection{Proof of \texorpdfstring{\cref{thm:ipa-nonamort}}{the Non-Amortizability Theorem}}
\label{sec:ipa-nonamort-proof}

\begin{proof}
The proof follows that of \cref{prop:bp-nonamort}.
In particular, the counting argument used there does not depend on how the
generators are sampled.
It bounds the number of descriptions of a straight-line program of $Q$ group
operations over $M$ handles by $(2M^2)^{Q}M^{2m}$, and the only property of
$M$ it uses is $\log M=O(\log\lambda)$.
A structured key, and any polynomial number of elements published with it,
only enlarge $M$.
The differences are the following.

\begin{itemize}[leftmargin=*]
\item \emph{Model.}
The adversary works in the generic bilinear group model and additionally has
$e:\mathbb G_1\times\mathbb G_2\to\mathbb G_T$.
The messages whose cost is bounded lie in $\mathbb G_1$, a pairing outputs an
element of $\mathbb G_T$, and no operation maps $\mathbb G_T$ back to
$\mathbb G_1$, so pairings do not weaken the count.
Handles in $\mathbb G_2$ and $\mathbb G_T$ enter only through $M$.

\item \emph{Coupling.}
The real and symbolic executions are coupled by bounding the probability that
distinct formal coefficient vectors evaluate to the same group element.
For uniform generators this is at most $1/(p-1)$ per pair.
Under $\mathcal{PW}_N$ a collision between $c\ne d$ makes the trapdoor a root of
$\sum_i(c_i-d_i)X^{i-1}$, a nonzero polynomial of degree at most $N-1$, so the
bound becomes $(N-1)/p$; under $\mathcal{ML}_N$ it is $r/p$ by
Schwartz--Zippel.
Both remain negligible after the union bound, since $\log p=\omega(\log\lambda)$
makes $p$ superpolynomial, and the assumption invoked is the Find-Rep assumption
of the binding sketch in \cref{app:construction}.

\item \emph{Trapdoor.}
The trapdoor is sampled by $\setup$ and never given to the adversary, and this
is necessary.
An adversary holding it computes every multi-scalar multiplication
$[\vec a^{\top}\mathbf r]_1$ with $\Theta(N)$ field operations and one
exponentiation, so every round of every proof costs $O(\log p)$ group
operations and the conclusion fails.

\item \emph{Prescribed coefficients and fresh challenges.}
The folding of \cref{fig:pv-ipa} differs from that of \cref{fig:bp-ipa}, but the
two facts the counting needs carry over.
First, the pairing checks of step~4 and the final check force the folded
generators $[r']_1,[s']_1$ to be the honest ones, and the updates of
$[\alpha]_1,[\beta]_1,z$ force the cross terms to have their prescribed formal
coefficients, each except with probability $O(1/p)$ per round by
Schwartz--Zippel in the challenge.
Second, fix the history before the challenge $c_i$ of round $i<r$.
Expanding the folds $\vec a'=c\,\vec a_L+c^2\vec a_R$ and
$[\mathbf r']_1=c^{-1}[\mathbf r_L]_1+c^{-2}[\mathbf r_R]_1$ shows that each
cross term of round $i+1$ has the form $A\,c_i+B+C\,c_i^{-1}$, with $A,B,C$
fixed by the history.
For instance $[\alpha'_{-1}]_1$ has $A=[\vec a_{RL}^{\top}\mathbf r_{LR}]_1$ and
$C=[\vec a_{LL}^{\top}\mathbf r_{RR}]_1$, where the double subscripts denote
quarters.
The cross terms are nonconstant in $c_i$ unless the current vectors vanish, and
the vector $\mathbf 1$ folds to a nonzero vector except with negligible
probability, so the $\beta$ terms are nonconstant even when the message is
zero.
Each equation $A\,c_i+B+C\,c_i^{-1}=z$ has at most two solutions, so the target
has min-entropy at least $\log(p-1)-1$ before $c_i$ is sampled.
This replaces \eqref{eq:ipa-lb-laurent} and the bound of four solutions in the
reconstruction step, which becomes a bound of two.
The folded generators $[r']_1,[s']_1$ sent after each challenge, including the
last, are further challenge-dependent messages.
We do not count them, so the overhead $(r-1)\log p/\log\lambda$ is conservative.

\item \emph{Signature.}
Signing adds group operations to each proof and removes none, so it does not
affect a lower bound on their number.

\item \emph{Support.}
The additive $N$ term of \cref{thm:ipa-lower-bound} does not transfer, since
precomputed elements of full support break its induction.
It is paid once per generator set and never enters $\Delta$.
\end{itemize}
Nothing else in the proof of \cref{prop:bp-nonamort} refers to the
distribution of the commitment key or to the folding rule, and the global
counting then gives $Q=\Omega(n\Delta)$ as before.
\end{proof}
 \section{Our Construction from Inner-Product Arguments and Signatures}
\label{app:construction}

This appendix gives the construction deferred from \cref{sec:ipa-construction}
in full.
We first state the syntax of each ingredient and the exact property we need
from it, and then give the construction.
The protocol is that of TRACE~\cite{trace}, modified as described in
\cref{sec:ipa-construction}.
We instantiate it with a Pedersen vector commitment and with the inner-product
argument of Daza, R\`afols and Zacharakis~\cite{drz} (\cref{fig:pv-ipa}), not
with that of Bulletproofs.

Throughout, $N=2^r$ is the padded vector length and $\ro\sample\roSet$ is the
random oracle of \cref{sec:mechanism}.
$\setup$ and $\KG$, and hence $\mathsf{Gen}$ and $\mathsf{Sig.KG}$, do not query
it.

\subsection{Ingredients and the Properties We Need}
\label{app:ingredients}

\paragraph{Structured keys and the Pedersen commitment.}
A key generator $\mathsf{Gen}(1^\lambda,N)$ samples asymmetric bilinear group
parameters $\mathsf{pp}=(p,\mathbb G_1,\mathbb G_2,\mathbb G_T,e,[1]_1,[1]_2)$
and two independent vectors $\mathbf r,\mathbf s\in\mathbb F_p^N$ from a key
distribution $\mathcal D_N$.
Examples are the powers distribution $\mathcal{PW}_N$, with
$\mathbf r=(1,x,\ldots,x^{N-1})$ for a uniform $x$, and the multilinear
distribution $\mathcal{ML}_N$ of~\cite{drz}.
It outputs a prover key
$\mathsf{ck}_{\mathsf P}=(\mathsf{pp},[\mathbf r]_1,[\mathbf s]_1,\mathsf{ck}_{\mathsf V})$
and a verifier key $\mathsf{ck}_{\mathsf V}$ holding $\mathsf{pp}$ and the
compressed descriptions $[\mathbf x]_2,[\mathbf y]_2$ of the two keys, of size
$O(\log N)$.
The Pedersen vector commitment to $\vec a\in\mathbb F_p^N$ is
$[\vec a^{\top}\mathbf r]_1$.

\begin{definition}[Find-Rep relative to $\mathsf{Gen}$]
\label{def:findrep}
For all PPT adversaries $\adv$ and all $\lambda\in\mathbb N$:
\[ \multilinePr{
    & \protect\vec a\neq\mathbf 0\ \wedge\ [\protect\vec a^{\top}\mathbf r]_1=[0]_1
}{
    & (\mathsf{ck}_{\mathsf P},\mathsf{ck}_{\mathsf V})\gets\mathsf{Gen}(1^\lambda,N)\\
    & \protect\vec a\gets\adv\left( \mathsf{ck}_{\mathsf P} \right)
} \leq \negl \]
\end{definition}

For $\mathcal{PW}_N$ and $\mathcal{ML}_N$ this is the asymmetric
$\mathcal D_N$-Find-Rep assumption of~\cite{drz}, which reduces to $q$-DLOG,
respectively DLOG, in the asymmetric bilinear group.

\paragraph{Labelled inner-product arguments.}
For a prover key $\mathsf{ck}_{\mathsf P}=(\mathsf{pp},[\mathbf r]_1,[\mathbf s]_1,\mathsf{ck}_{\mathsf V})$,
the inner-product relation is
\[
  \rel_{\mathsf{IP}}[\mathsf{ck}_{\mathsf P}]
  =\left\{
    \bigl(([\alpha]_1,[\beta]_1,z),\,(\vec a,\mathbf b)\bigr)
    \;\middle|\;
    [\alpha]_1=[\vec a^{\top}\mathbf r]_1,\;
    [\beta]_1=[\mathbf b^{\top}\mathbf s]_1,\;
    \langle\vec a,\mathbf b\rangle=z
  \right\},
\]
with $\vec a,\mathbf b\in\mathbb F_p^N$ and $z\in\mathbb F_p$.
This is the relation $\mathcal L_{\mathrm{IP}}$ of \cref{fig:pv-ipa} with the
keys fixed by $\mathsf{ck}_{\mathsf P}$, so its key-structure conditions
$\mathcal L_{\mathsf{Com}}$ hold by construction of $\mathsf{Gen}$.
A labelled inner-product argument is a pair of oracle algorithms
$\left( \mathsf P,\mathsf V \right)$ over the keys of $\mathsf{Gen}$:
\begin{itemize}
\item $\mathsf P^\ro(\mathsf{ck}_{\mathsf P},\ell,X,w)$ outputs a proof
  $\pi_{\mathsf{ip}}$ for the statement $X$ with witness $w$ under the label
  $\ell\in\bits^*$;
\item $\mathsf V^\ro(\mathsf{ck}_{\mathsf V},\ell,X,\pi_{\mathsf{ip}})$ outputs a
  bit.
\end{itemize}
Both query the random oracle only on inputs that begin with $\ell$, so arguments
under distinct labels use independent oracles.
We need three properties.

\begin{definition}[Properties of the inner-product argument]
\label{def:ip-props}
\leavevmode
\begin{enumerate}[label=(IP\arabic*),leftmargin=*]
\item \emph{Perfect completeness.}
For all $(\mathsf{ck}_{\mathsf P},\mathsf{ck}_{\mathsf V})$ in the support of
$\mathsf{Gen}(1^\lambda,N)$, all $\ro\in\roSet$, all $\ell\in\bits^*$ and all
$(X,w)\in\rel_{\mathsf{IP}}[\mathsf{ck}_{\mathsf P}]$,
$\mathsf V^\ro(\mathsf{ck}_{\mathsf V},\ell,X,
 \mathsf P^\ro(\mathsf{ck}_{\mathsf P},\ell,X,w))=1$.
\label{ip:completeness}
\item \emph{Black-box knowledge soundness.}
There is an expected polynomial-time extractor $\mathcal E$ such that for all
PPT adversaries $\adv$ and all $\lambda\in\mathbb N$:
\[ \multilinePr{
    & \mathsf V^\ro(\mathsf{ck}_{\mathsf V},\ell,X,\pi_{\mathsf{ip}})=1\ \wedge\\
    & \quad (X,w)\notin\rel_{\mathsf{IP}}[\mathsf{ck}_{\mathsf P}]
}{
    & \ro\sample\roSet;\ (\mathsf{ck}_{\mathsf P},\mathsf{ck}_{\mathsf V})\gets\mathsf{Gen}(1^\lambda,N)\\
    & \rho\sample\bits^{\poly}\\
    & (\ell,X,\pi_{\mathsf{ip}})\gets\adv^\ro(\mathsf{ck}_{\mathsf P};\rho)\\
    & w\gets\mathcal E^{\adv,\ro}(\mathsf{ck}_{\mathsf P},\ell,X,\rho)
} \leq \negl \]
where $\mathcal E$ may run $\adv$ on random tapes of its choice and answer each
of its random oracle queries with a value of its choice.
\label{ip:ks}
\item \emph{Succinctness.}
$\abs{\mathsf{ck}_{\mathsf V}}$ and $\abs{\pi_{\mathsf{ip}}}$ are
$\poly\cdot O(\log N)$, and $\mathsf V$ runs in time
$\poly\cdot O\left( \abs{\ell}+\log N \right)$.
\label{ip:succinct}
\end{enumerate}
\end{definition}

\paragraph{Signatures.}
A signature scheme
$\Sigma=\left( \mathsf{Sig.KG},\mathsf{Sig.Sign},\mathsf{Sig.Vfy} \right)$ has
$\mathsf{Sig.KG}(1^\lambda)\to(\mathsf{ssk},\mathsf{svk})$,
$\mathsf{Sig.Sign}^\ro(\mathsf{ssk},m)\to\sigma$ and
$\mathsf{Sig.Vfy}^\ro(\mathsf{svk},m,\sigma)\to\bits$.

\begin{definition}[Properties of the signature scheme]
\label{def:sig-props}
\leavevmode
\begin{enumerate}[label=(S\arabic*),leftmargin=*]
\item \emph{Correctness.}
For all $(\mathsf{ssk},\mathsf{svk})$ in the support of $\mathsf{Sig.KG}(1^\lambda)$,
all $\ro\in\roSet$ and all $m\in\bits^*$,
$\mathsf{Sig.Vfy}^\ro(\mathsf{svk},m,\mathsf{Sig.Sign}^\ro(\mathsf{ssk},m))=1$.
\label{sig:correct}
\item \emph{Existential unforgeability under chosen-message attack.}
For all PPT adversaries $\adv$ and all $\lambda\in\mathbb N$:
\[ \multilinePr{
    & \mathsf{Sig.Vfy}^\ro(\mathsf{svk},m^*,\sigma^*)=1\ \wedge\\
    & \quad m^*\notin\mathcal Q_\Sigma
}{
    & \ro\sample\roSet\\
    & (\mathsf{ssk},\mathsf{svk})\gets\mathsf{Sig.KG}(1^\lambda)\\
    & (m^*,\sigma^*)\gets\adv^{\ro,\,\mathsf{Sig.Sign}^\ro(\mathsf{ssk},\cdot)}(\mathsf{svk})
} \leq \negl \]
where $\mathcal Q_\Sigma$ is the set of messages $\adv$ queried to its signing
oracle.
\label{sig:euf}
\end{enumerate}
\end{definition}

\subsection{The Construction}
\label{app:construction-formal}

\Cref{fig:construction} gives $\Pi_{\mathsf{IPA}}$ in full.
The message space is $\mathcal M_\lambda=\mathbb F_p^N$ of traces.
The inner-product argument is run on the trace and the public all-ones vector,
so the only secret a party holds is its signing key.

\begin{figure}[tp]
\centering
\fbox{\begin{minipage}{\dimexpr\linewidth-2\fboxsep-2\fboxrule\relax}
\small
\begin{minipage}[t]{0.47\linewidth}
\underline{$\setup(1^\lambda)$}\\
$(\mathsf{ck}_{\mathsf P},\mathsf{ck}_{\mathsf V})\gets\mathsf{Gen}(1^\lambda,N)$\\
$[S]_1\gets\sum_{i=1}^{N}[s_i]_1$\\
return $\pcrs=(\mathsf{ck}_{\mathsf P},[S]_1)$, $\vcrs=(\mathsf{ck}_{\mathsf V},[S]_1)$
\\[6pt]
\underline{$\KG(\pcrs)$}\\
$(\mathsf{ssk},\mathsf{svk})\gets\mathsf{Sig.KG}(1^\lambda)$\\
return $\sk=\mathsf{ssk}$, $\pk=\mathsf{svk}$
\\[6pt]
\underline{$\Com(\pcrs,E)$}\\
return $[E^{\top}\mathbf r]_1$
\end{minipage}\hfill
\begin{minipage}[t]{0.50\linewidth}
\underline{$\prv^\ro(\pcrs,\sk,E,\nonce)$}\\
$\cm\gets\Com(\pcrs,E)$;\ $v\gets\langle E,\mathbf 1\rangle$\\
$X\gets(\cm,[S]_1,v)$\\
$\theta\gets\ro(\mathsf{claim},\cm,\pk,v,\nonce)$\\
$\pi_{\mathsf{ip}}\gets\mathsf P^\ro(\mathsf{ck}_{\mathsf P},\theta,X,(E,\mathbf 1))$\\
$\sigma\gets\mathsf{Sig.Sign}^\ro(\mathsf{ssk},(\cm,\nonce,v,\pi_{\mathsf{ip}}))$\\
return $\pi=(v,\pi_{\mathsf{ip}},\sigma)$
\\[6pt]
\underline{$\vfy^\ro(\vcrs,\pk,\cm,\nonce,\pi)$}\\
parse $\pi=(v,\pi_{\mathsf{ip}},\sigma)$ and $\pk=\mathsf{svk}$\\
if $\mathsf{Sig.Vfy}^\ro(\mathsf{svk},(\cm,\nonce,v,\pi_{\mathsf{ip}}),\sigma)=0$: return $0$\\
$X\gets(\cm,[S]_1,v)$\\
$\theta\gets\ro(\mathsf{claim},\cm,\pk,v,\nonce)$\\
return $\mathsf V^\ro(\mathsf{ck}_{\mathsf V},\theta,X,\pi_{\mathsf{ip}})$
\end{minipage}
\end{minipage}}
\caption{The argument with registered provers $\Pi_{\mathsf{IPA}}$, built from
a labelled inner-product argument $(\mathsf{Gen},\mathsf P,\mathsf V)$ and a
signature scheme $\Sigma$.
The key enters the inner-product argument only through the label $\theta$,
which hashes $\pk$, and the verifier receives no opening of $\cm$.}
\label{fig:construction}
\end{figure}

The statement is well formed.
The keys are well structured by construction of $\setup$, and
\[
    [E^{\top}\mathbf r]_1=\cm,
    \qquad
    [\mathbf 1^{\top}\mathbf s]_1=[S]_1,
    \qquad
    \langle E,\mathbf 1\rangle=v,
\]
so $(X,(E,\mathbf 1))\in\rel_{\mathsf{IP}}[\mathsf{ck}_{\mathsf P}]$.
The eligibility test of \cref{alg:both} hashes $\pi_{\mathsf{ip}}$, which is
a deterministic function of the trace, the key and the block, and not $\sigma$,
so that a randomized signature cannot be resampled to win.

\subsection{Properties}
\label{app:construction-properties}

We sketch why $\Pi_{\mathsf{IPA}}$ satisfies the properties of \cref{def:dp-arg}.

\emph{Completeness} follows directly from the well-formedness above, from
perfect completeness of $(\mathsf{Gen},\mathsf P,\mathsf V)$~\cite[Thm.~4]{drz},
and from correctness of $\Sigma$.
The label $\theta$ separates the oracle queries of different claims, so the
Fiat--Shamir compiled argument is complete for the honestly derived statement.

\emph{Binding} is binding of the Pedersen vector commitment under the structured
key.
Two openings $E\ne E'$ of one commitment give the nonzero vector
$E-E'$ with $[(E-E')^{\top}\mathbf r]_1=[0]_1$, which breaks
the $\mathcal{PW}_N$-Find-Rep assumption, respectively the
$\mathcal{ML}_N$-Find-Rep assumption.
These reduce to $q$-DLOG, respectively DLOG, in the asymmetric bilinear
group~\cite{drz}.

\emph{Succinctness} follows directly from the properties of the Pedersen
commitment and of the inner-product argument.
The commitment is a single element of $\mathbb G_1$.
The proof consists of $v$, of $\pi_{\mathsf{ip}}$, which has $O(\log N)$ elements
of $\mathbb G_1$ and $\mathbb F_p$~\cite[Thm.~4]{drz}, and of one signature.
The verifier checks one signature, hashes the nonce into $\theta$, and runs
$\mathsf V$, which reads the $O(\log N)$ elements of $\mathsf{ck}_{\mathsf V}$
and performs $O(\log N)$ group operations and pairings.

\emph{Knowledge soundness} follows from the analysis of Daza et
al.~\cite[Thm.~4]{drz}, which shows that their argument has witness-extended
emulation under the Find-Rep assumption above.
The emulation is black-box: the extractor rewinds the prover to collect a tree of
accepting transcripts over fresh challenges, first recovers the folded keys and
then, as in the extractor of Bootle et al.~\cite{bootle2016ipa}, the openings.
After Fiat--Shamir, the extractor obtains the fresh challenges by reprogramming
the random oracle under the label $\theta$, and the analysis of Fiat--Shamir for
multi-round special-sound protocols~\cite{TCC:AttFehKlo22} preserves a
black-box extractor with a polynomial loss.
Given an adversary that outputs an accepting $(\cm,\nonce,\pi)$, the extractor
therefore recovers a witness whose first component $E$ satisfies
$[E^{\top}\mathbf r]_1=\cm$, and $E$ is an opening of $\cm$ as \cref{def:ks}
requires.
The signature plays no part here, since the adversary of that definition holds
$\sk$ and can sign every transcript the rewinding produces.

\emph{Designated proving} follows from unforgeability of $\Sigma$.
Let $\adv$ output an accepting $(\cm,\nonce,\pi)$ with $(\cm,\nonce)\notin\mathcal Q$.
Every proof returned by $\prv^\ro_\sk$ carries a signature on a message that
begins with a pair $(\Com(\pcrs,E'),\nonce')\in\mathcal Q$, so the message
$(\cm,\nonce,v,\pi_{\mathsf{ip}})$ signed by $\sigma$ was never signed by the
oracle, and $\sigma$ is a forgery.
The reduction receives $\mathsf{svk}$ and a signing oracle, and answers each
query to $\prv^\ro_\sk$ by computing $v$ and $\pi_{\mathsf{ip}}$ honestly,
which requires no secret, and asking its oracle for $\sigma$.
Its simulation is perfect, and designated proving rests on the signature
alone.
 \section{Costing a TRACE Claim on Arbitrum One}
\label{app:costing}

This appendix carries the hardware, the benchmarks, and the arithmetic behind the
costs reported in \cref{sec:example}.

\paragraph{Node cost.}
An Arbitrum full node holds a snapshot of $2.4$\,TB and needs headroom above
it. So the minimum hardware for a full node is eight cores, 64\,GB of memory, and 4\,TB of
NVMe storage.
A managed provider charges about a dollar an hour for such a
node~\cite{chainstack}, so
\begin{equation}
\label{eq:example_c}
  c \;=\; \$1.00 \text{ per hour} \;=\; \$730 \text{ per month},
\end{equation}
where $730$ is the number of hours in an average month.
Self-hosting the equivalent machine costs about \$0.41 per hour and would raise
the bound of \cref{eq:bound} from two attention nodes to five.

\paragraph{Cost of one attention sybil.}
An attention proof is calculated once per block, and Arbitrum produces a block every $0.25$
seconds. So, an identity calculates 14,400 proofs per hour.
Producing one TRACE proof costs about $4N$ points of multi-scalar multiplication,
where $N$ is the trace length.
Trace length is a design choice for the rollup operator. A shorter trace lowers
$d_a$, the proving cost per identity, but also lowers the sustainable execution
diversity:
\[
  A \leq 1 + \frac{2d_a}{c}.
\]
A longer trace raises both the cost of deploying identities and the bound on
$A$. The rollup operator therefore trades proving work against execution
diversity.
Half of that is the trace-side commitments and half is the folding of the
commitment key, each of which sums to $2N$
(see \cref{app:cost}).
One claim is therefore four multi-scalar multiplications of size $N$.

\Cref{tab:msm} gives measured times for one such multiplication on BN254.

\begin{table}[t]
\centering
\caption{Time for one multi-scalar multiplication of $N$ points on BN254,
measured on a 12-core AMD Ryzen 9 9900X. The time is close to linear in $N$,
since the window term in Pippenger's algorithm moves slowly over this range.}
\label{tab:msm}
\begin{tabular}{rr}
\toprule
$N$ & time \\
\midrule
$2^{18}$ & 50\,ms \\
$2^{19}$ & 100\,ms \\
$2^{20}$ & 180\,ms \\
$2^{21}$ & 360\,ms \\
\bottomrule
\end{tabular}
\end{table}

At $N = 2^{20}$ one proof costs $4 \times 180 = 720$\,ms of that machine, and
an identity therefore consumes
\[
  \frac{14{,}400 \times 0.72}{3600} \;=\; 2.9 \text{ machine-hours per hour},
\]
where the numerator is the machine-seconds spent in an hour and $3600$ converts
to hours.
Therefore, an identity needs about three of these machines running without pause.

We price that machine as the closest AWS instance by thread count, a
\texttt{c8a.6xlarge} with 24 virtual cores, at \$226 per month under a
three-year savings plan, or \$0.31 per hour.
For illustration, a one-ETH registration deposit costs \$0.014 per hour at
\$2{,}500 per ETH and the same interest rate.
The per-identity cost is therefore
\begin{equation}
\label{eq:example_d_a}
  d_a \;=\; 2.9 \times 0.31 + 0.014 \;=\; \$0.91 \text{ per hour},
\end{equation}
where the first term is proving and the second is the deposit. The latter is a negligibile cost and
the exact deposit size is not very important.

Trace length is the only parameter of the mechanism that moves $d_a$, and
\cref{tab:sweep} sweeps it.

\begin{table}[t]
\centering
\caption{Cost per identity and the sustainable number of attention nodes as the
trace length varies. Instances are \texttt{c8a.6xlarge} at \$0.31 per hour and
$c = \$1.00$ per hour throughout. 29 attention nodes requires 46
instances per identity, which is \$10{,}400 per month of proving.}
\label{tab:sweep}
\begin{tabular}{|r|r|r|r|r|}
\toprule
$N$ & per claim & instances & $d_a$ per hour & $\lfloor A \rfloor$ \\
\midrule
$2^{18}$ & 0.20\,s  & 0.8  & \$0.26  & 1  \\
$2^{20}$ & 0.72\,s  & 2.9  & \$0.91  & 2  \\
$2^{21}$ & 1.44\,s  & 5.8  & \$1.81  & 4  \\
$2^{22}$ & 2.88\,s  & 11.6 & \$3.61  & 8  \\
$2^{24}$ & 11.5\,s  & 46   & \$14.27 & 29 \\
\bottomrule
\end{tabular}
\end{table}

Lowering the rate at which claims win does not reduce $d_a$ under TRACE.
Every identity proves on every block whether or not it wins. The winning rate
sets how often a claim is submitted, not how often a proof is computed.
In Proof of Diligence, the keyed-output test samples which blocks require a
proof. This reduces $d_a$, and therefore lowers $A$ rather than raising it.
TRACE proves on every block, so this sampling option is unavailable in the
main instantiation.

\paragraph{The price of diversity.}
Another way to look at \cref{eq:bound} is that each additional attention node requires
$d_a$ to rise by $c/2$.
At these above costs that is about three more instances per identity for every two
nodes added. The work those instances do is useless, since it exists to make a sybil expensive
rather than to produce any useful computation.
Twenty-nine attention nodes therefore costs \$10{,}400 per month per identity,
and \$301{,}000 per month across the chain.

The above price follows from insisting on non-decreasing reward functions.
Under such a reward, an operator gains from adding identities, so the only way
to stop it is to make each identity expensive. A reward that decreases past the
target changes this trade-off. It can discourage extra identities without
increasing $d_a$, but it shifts the constraint to the reward curve and the
target population.

\subsection{Reward-curve tradeoff}
\label{app:reward-tradeoff}

For the example in \cref{sec:tradeoff}, the total reward pool is $\$90$ and the
target is $30$ identities. With one identity per operator, each nets
$3-0.91-1=1.09$. A deviating operator with two identities nets
$2\times3-2\times0.91-1=3.18$. If it drops one, the count falls to $29$, the
pool rises to $\$145$, and its payoff is $5-0.91-1=3.09$. If it adds one, its
payoff is $3\times1-3\times0.91-1=-0.73$, while a new operator entering gets
$1-0.91-1=-0.91$. Thus no deviation is profitable, and the chain receives half
the number of operators.

Dropping an identity raises every remaining identity's payment by $\beta$, so an
operator holding $q$ identities gains $\beta(q-1)$ and loses the dropped
identity's payment. Adding a second identity at $q=1$ leaves both earning
$p_a-\beta$, so one identity is preferable when $\beta>p_a/2$. Dropping one at
$q=2$ takes the count to $A-1$ and pays the remaining identity $\beta$ more, so
dropping is preferable when $\beta>p_a$. More generally, dropping is preferred
when $\beta>p_a/(q-1)$, so any $\beta>p_a$ makes $q>1$ unstable.

The reward pool reaches zero at $A+(d_a+p_a)/\beta$. At $\beta=2$, one extra
identity puts every other operator at $-\$0.91$, and two extra identities exhaust
the pool. Let $m=(d_a+p)/\beta$. Since
$\beta>\max\{p/2,p-c\}$, we have
\[
  m<f(p_a),\qquad
  f(p_a):=\frac{d_a+p_a}{\max\{p_a/2,p_a-c\}}.
\]
The two terms in the maximum cross at $p=2c$, because
$p/2\geq p-c$ exactly when $p\leq2c$. Below that point,
$f(p)=2+2d_a/p$, which decreases with $p$. Above it,
$f(p)=(d_a+p)/(p-c)$, whose derivative is
$-(c+d_a)/(p-c)^2<0$. The two expressions agree at $p=2c$, so $f$ decreases
throughout $p>c$ and is largest as $p$ approaches $c$. Thus the gap above $A$
is at most $2+2d_a/c$, which is \cref{eq:bound} again.

\section{Proof of Proposition~\ref{prop:necessity_of_condition}}
\label{app:prop_necessity_of_condition}

\begin{proof}
    Let $p:=\frac{g_a(A)}{A}-d_a$ be the attention game player's profit, ignoring entry costs. Individual rationality requires $p\geq c$. Let $\delta:=g_a(A+1)-g_a(A)$. Since the current player should prefer having one sybil to having $2$ sybils, we have the following inequality: 

    \begin{equation}
        p_a\geq \frac{2g_a(A+1)}{A+1}-2d_a.
    \end{equation}

    Substituting $g_a(A)=A(d_a+p)$ (from the definition of $p$) and $g_a(A+1)=g_a(A)+\delta$ simplifies the inequality to 

    \begin{equation}
        p_a(A-1)+2\delta \leq 2d_a, 
    \end{equation}
    Therefore, for non-decreasing $g_a$, we have $\delta\geq 0$, and 

    \begin{equation}
        c<\frac{2d_a-2\delta}{A-1}\leq \frac{2d_a}{A-1}.
    \end{equation}
    \qed
\end{proof}

\section{Equilibrium with Multiple Attention Sybils}
\label{app:multi-sybil}

The condition~\eqref{condition_for_good_equilibrium} can be weakened when each
attention player is allowed to deploy more than one sybil. Suppose each of $M$
players deploys $s_a$ sybils, so that $M s_a=A$. We allow fractional values for
players or sybils in this calculation. The necessary condition is
\begin{equation}
\label{condition_sa_sybils}
  c < \frac{d_a s_a(s_a+1)}{A-s_a}
    = \frac{d_a(s_a+1)}{M-1}.
\end{equation}
For $s_a=1$, this is exactly
\eqref{condition_for_good_equilibrium}.

\begin{proposition}
\label{prop:necessity_sa_sybils}
Condition~\eqref{condition_sa_sybils} is necessary for an equilibrium with $M$
players each having $s_a$ sybils in the attention game.
\end{proposition}

\begin{proof}
Let $P_a$ be each player's profit at the target,
$P_a:=s_a\left(\frac{g_a(A)}{A}-d_a\right)$. Individual rationality requires
$P_a>c$. Since $g_a$ is non-decreasing,
$g_a(A+1)\geq g_a(A)$. If a player adds one sybil, its payoff is at least
$(s_a+1)\left(\frac{g_a(A)}{A+1}-d_a\right)$. Requiring this to be no more than
$P_a$ gives $P_a(M-1)\leq d_a(s_a+1)$. Combining this with $P_a>c$ yields
\[
  c < \frac{d_a(s_a+1)}{M-1},
\]
which is equivalent to~\eqref{condition_sa_sybils}.
\end{proof}
 \section{A preprocessing variant: why key derivation alone is not enough}
\label{app:keyed}

This appendix specifies the strengthening discussed in
\cref{sec:variant-limitation}, and gives the preprocessing argument that limits
what it delivers.

Its cause is the vector the trace is folded against.
The main construction folds the trace against the all-ones vector, which is the
same for every identity, so the vector side of each round costs nothing per
identity.
A natural strengthening derives $\mathbf b$ from the registered key, fixes it at
registration, and certifies once with a succinct zero-knowledge argument that the
derivation was carried out correctly.
The claim algebra changes only in where the verifier finds the commitment to
$\mathbf b$: in the registered key rather than in the reference string.
Each round then folds a vector belonging to one identity alone.

That instantiation does not deliver a recurring cost proportional to the trace.
Since $\mathbf b$ is fixed at registration, the round-$t$ cross term
$[\mathbf b_L^{\top}\mathbf s_R]_1$ of \cref{fig:pv-ipa} decomposes into sums
\[
  T_{\alpha\beta}
  =\Bigl[\sum_j \mathbf b_{\tau(\beta,j)}^{\top}\mathbf s_{\sigma(\alpha,j)}\Bigr]_1,
\]
where $\sigma$ and $\tau$ index the blocks of the key $[\mathbf s]_1$ and of the
secret vector that round $t$ combines.
The challenge coefficients factor out of these products.
They therefore do not depend on the block, and an operator computes them once per
identity.
Preprocessing to depth $k$ costs $N(2^{k}-1)/2$ group operations and stores
$(4^{k}-1)/3$ elements, after which the recurring secret-side work is
$(4^{k}-1)/3+N/2^{k-1}$.
Balancing the two terms gives $k=\left( \log_2N+1 \right)/3$ and a recurring
cost of $\tilde O\left( N^{2/3} \right)$.
At $N=2^{20}$ and $k=7$, an operator stores about $5{,}500$ group elements for a
one-time $6.7\times10^{7}$ group operations, and the recurring secret-side work
falls from $2\times10^{6}$ to about $2\times10^{4}$.

The work is not removed.
It is relocated.
Each table belongs to one identity, so an operator holding many pays for each
separately, and the cost that cannot be amortized across identities survives as a
one-time $\tilde O\left( N^{4/3} \right)$ charge at registration.
Reporting that as an hourly figure requires an assumption about how long an
identity stays registered and a benchmark for the registration proof.
This paper supplies neither.
Whether the folding order can be made to depend on the block, which would restore a
recurring cost, is open.
 
\section{Per-Identity Cost of a TRACE Claim}
\label{app:cost}

The claim of \cref{alg:both} rests on an inner-product argument between the trace
$E$ and a vector $\mathbf b$, which in our variant is the all-ones
vector.
As in \cref{app:construction}, the commitment is a Pedersen commitment and the
argument is that of Daza, R\`afols and Zacharakis~\cite{drz}, not that of
Bulletproofs.
Write $N=2^r$ for the padded trace length, so the argument runs for $r$
rounds.
Each round halves both vectors.
The prover sends commitments to the halves of each vector together with two
cross-term scalars, the inner product of the low half of $E$ with the high
half of $\mathbf b$ and the inner product of the high half of $E$ with the low
half of $\mathbf b$.
A challenge derived from those messages folds each vector into one of half the
length, and the next round repeats on the result.

The commitments on the trace side are multi-scalar multiplications over the
round's vector, so the scalar multiplications across all rounds sum to
$N + N/2 + N/4 + \cdots < 2N$.
The vector side is cheap.
Folding keeps $\mathbf b$ a scalar multiple of the all-ones vector, so its cross
commitments are combinations of sums of dyadic blocks of $[\mathbf s]_1$, which
an operator computes once from the reference string.
The cross terms are inner products over the round's vector, so they contribute
fewer than $2N$ field operations in total.
One claim therefore costs $\Theta(N)$.

One cost is left out of the count above.
Each round also folds the commitment key $[\mathbf r]_1$, which the prover
computes itself, and
that is work of the same order.

The count is what an honest prover pays for one claim.
An operator holding many identities can share part of it, at best reaching
$N^{0.936}$ per identity (\cref{app:amortization}), so the cost it faces stays
close to linear, which is what \cref{sec:example} prices.

\section{How Much Proving Work an Operator Can Share}
\label{app:amortization}

\Cref{app:cost} prices a claim at the cost the honest prover pays.
This appendix asks how much of it an operator running many identities can avoid.
The saving is small.

Some sharing is immediate.
Every claim folds the trace against the same all-ones vector, so in the
argument of \cref{fig:pv-ipa} the whole first round is independent of the key:
\[
  [\beta_{-1}]_1=\sum_{i=N/2+1}^{N}[s_i]_1,
  \qquad
  [\beta_{1}]_1=\sum_{i=1}^{N/2}[s_i]_1,
\]
and the trace-side messages $[\alpha_{\pm1}]_1$ depend on the trace and the
commitment key alone.
An operator computes the first round once for all of its identities, which first
diverge at the first challenge, the point where the key enters through the
Fiat--Shamir label.

The sharing can be pushed further.
After $t$ rounds each message is built from $2^t$ pieces of the trace and $2^t$
pieces of the commitment key, with that identity's challenges supplying the
coefficients.
An operator can compute those pieces once, then evaluate them for each identity
against its own challenges in protocol order, so no challenge is ever guessed.
Fix a batch width $B=2^k$ with $B\le\sqrt N$.
After $k$ rounds every identity's folded vector is a combination of the same $B$
pieces, so producing all $n$ of them at once is a matrix product of shape
$(n\times B)(B\times N/B)$, which fast matrix multiplication computes in fewer
than the obvious $nN$ operations~\cite{strassen}.
The formulas need addition in each operand and multiplication only between them,
so field scalars acting on group elements are admissible.
Each proof then finishes on its own over vectors of length $N/B$.

Counting field operations, group additions, and scalar multiplications at unit
cost, the total is
\[
  T(n,N) \;=\; O\!\left(NB \;+\; nB^2 \;+\; nNB^{\omega-3} \;+\; nN/B\right),
\]
where $\omega$ is the matrix multiplication exponent, and the four terms are the
shared precomputation, the per-identity evaluation, the batched folding, and the
independent tails.
Taking $n=N$ and $B=N^{1/3}$ gives $O(N^{1+\omega/3})$ overall, or
\[
  O\!\left(N^{\omega/3}\right) \;=\; O\!\left(N^{0.936}\right)
\]
per identity at Strassen's exponent $\omega=\log_2 7$, where the saving against
the honest prover's $\Theta(N)$ is a factor $N^{1-\omega/3}$.
No choice of algorithm removes the dependence.
At the theoretical floor $\omega=2$ the cost per identity is still $N^{2/3}$.

Three costs are now known.
The proven lower bound is logarithmic in the trace length
(\cref{thm:ipa-nonamort}).
The best sharing we know costs $N^{0.936}$ per identity.
The honest prover costs $N$.
Sharing therefore saves a factor of $N^{0.064}$, which is $2.4$ at
$N=2^{20}$.
The cost per identity stays close to what the honest prover pays.

\Cref{sec:example} uses the honest prover's cost, measured on hardware, and we do
not reduce that figure.
The construction needs as many identities as the trace is long.
Its preprocessing has not been implemented or
measured~\cite{bensonballard}.
 
\section{A Construction with Stronger Non-Amortizability Guarantees}
\label{sec:prf-cons}

This construction is a variant of our main instantiation in \cref{sec:ipa-construction}. We discuss this only informally to show the feasibility of obtaining stronger guarantees; at the same time this design features a more expensive prover, which may limit its applicability and represents the primary trade off in comparison to our other scheme.

If the main underlying building block in our main scheme was an inner-product argument with some non-amortizability properties, here  we assume the same building block, plus a general-purpose argument (for any NP relation).

\paragraph{The construction.}
Let $F$ be a PRF with outputs in $\mathbb F_p$, and let $\Pi_{\mathsf{gp}}$ be a
general-purpose commit-and-prove\footnote{In the sense of~\cite{CCS:CamFioQue19}.} zero-knowledge argument whose
commitments are Pedersen commitments under the structured key $[\mathbf s]_1$ of
\cref{app:construction}.
Setup, the commitment $\cm$ to the trace, and the signature are as in
\cref{app:construction}, with the parameters of $\Pi_{\mathsf{gp}}$ added to the
reference string.
\begin{itemize}
\item At key generation the party samples a PRF key $k$ and a signature key
  pair, and registers a binding commitment to $k$ together with the
  verification key as $\pk$.
\item To claim, the prover fixes the context
  $\mathsf{ctx}=(\cm,\pk,\nonce)$ and derives a secret vector
  $\mathbf b\in\mathbb F_p^N$,
  \[
    b_i=F_k(i\,\|\,\mathsf{ctx})\quad(i\in[N]),
  \]
  in place of the all-ones vector of \cref{app:construction}.
  It computes $C_{\mathbf b}=[\mathbf b^{\top}\mathbf s]_1$ and
  $v=\langle E,\mathbf b\rangle$.
\item The proof consists of $C_{\mathbf b}$ and $v$, a proof
  $\pi_{\mathsf{PRF}}$, a proof $\pi_{\mathsf{IPA}}$, and a signature.
  The first proof is produced with $\Pi_{\mathsf{gp}}$ and shows that
  $C_{\mathbf b}$ commits to $(F_k(1\|\mathsf{ctx}),\ldots,F_k(N\|\mathsf{ctx}))$
  for the $k$ that $\pk$ commits to.
  The second is produced exactly as in \cref{app:construction}, with
  $C_{\mathbf b}$ in place of $[S]_1$: the statement is
  $X=(\cm,C_{\mathbf b},v)$ of
  \cref{fig:construction}, the witness is $(E,\mathbf b)$, and the Fiat--Shamir
  seed is $\theta=\ro(\mathsf{claim},\cm,\pk,v,C_{\mathbf b},\nonce)$.
  The signature covers all of the above, as in \cref{app:construction}.
\item The verifier checks the signature and $\pi_{\mathsf{PRF}}$, rebuilds $X$
  from $(\cm,C_{\mathbf b},v)$, and checks $\pi_{\mathsf{IPA}}$.
\end{itemize}
Deriving $\mathbf b$ from $k$ rather than sampling it keeps the claim a deterministic function of the block and the registered
key, so the prover cannot grind its eligibility.

\paragraph{Non-amortizability.}
In $\Pi_{\mathsf{IPA}}$ no key enters the witness, which is $(E,\mathbf 1)$.
Everything that depends on the trace is then shared across identities, which is the sharing exploited in \cref{app:amortization}, and
the only work that a fresh key makes fresh comes from the $r-1$ Fiat--Shamir
challenges, each worth $\log p$ bits.
This in turn yields the non-amortizability overhead $\Delta=(r-1)\log p/\log\lambda$ of
\cref{thm:ipa-nonamort}.

Here each identity carries $N$ fresh field elements per block instead.
Write $\mathbf b^{(j)}\in\mathbb F_p^N$ for the secret vector of identity $j$
and $b^{(j)}_i$ for its entries, following the convention of the proof of
\cref{thm:ipa-nonamort}, where $x^{(j)}_i$ is the $i$-th challenge of identity $j$.
The prover has to output $C_{\mathbf b^{(j)}}$, guaranteed to have high computational entropy by the soundness of $\pi_{\mathsf{PRF}}$
and binding of the commitment\footnote{More precisely, we need PRFs whose output is pseudorandom even for adversarially chosen keys. These can be realized in standard ways. See for example \cite[Appendix C]{EC:DGKR18}.}.
The global counting argument of the proof of \cref{thm:ipa-nonamort} then
applies with the vectors $\mathbf b^{(j)}$ in the role of the challenges.
A description of $Q$ group operations has $O(Q\log\lambda)$ bits and determines
every $C_{\mathbf b^{(j)}}$, hence all of $\mathbf b^{(1)},\ldots,\mathbf b^{(n)}$, which
together carry $nN\log p$ bits.
Informally, this gives
\[
  Q=\Omega\!\left(\frac{nN\log p}{\log\lambda}\right),
  \qquad\text{that is,}\qquad
  \Delta_{\mathsf{PRF}}=\frac{N\log p}{\log\lambda},
\]
the bound of \cref{thm:msm-entropy-lower-bound} with $\kappa=N\log p$, paid once
per identity.

\end{document}